%% file: jsdm-user-grouping-revised-3.tex
\documentclass[11pt,draftclsnofoot,peerreviewca,onecolumn,letterpaper]{IEEEtran}
\usepackage{amsmath,amsthm,amssymb,dsfont,stfloats,color,url}
\usepackage[pdftex]{graphicx}
\usepackage{subfigure}
\usepackage{cite}
\newtheorem{thm}{Theorem}

\newtheorem{lem}{Lemma}

\newtheorem{rem}{Remark}
\newtheorem{example}{Example}

\def\p0{{\pmb 0}}

\input{macros}

\IEEEoverridecommandlockouts

\begin{document}

\title{Joint Spatial Division and Multiplexing:
Opportunistic Beamforming and User Grouping}

\author{\authorblockN{Ansuman Adhikary and Giuseppe Caire}
\authorblockA{Ming-Hsieh Department of Electrical Engineering, University of Southern California, CA}
\thanks{This work was supported by the IT R\&D program of MKE/KEIT in Korea [Development of beyond 4G technologies
for smart mobile services].}}

\maketitle

\begin{abstract}
Joint Spatial Division and Multiplexing (JSDM) is a recently proposed scheme to enable massive MIMO like gains and simplified system operations for Frequency Division Duplexing (FDD) systems. The key idea lies in partitioning the users into groups with approximately similar covariances, and use a two stage downlink beamforming: a pre-beamformer that depends on the channel covariances and minimizes interference across groups and a multiuser MIMO precoder for the effective channel after pre-beamforming, to counteract interference within a group.
We first focus on the regime of a fixed number of antennas and large number of users, and show that opportunistic beamforming with user selection yields significant gain, and thus, channel correlation may yield a capacity improvement over the uncorrelated ``isotropic'' channel result of \cite{sharif2005capacity}. We prove that in the presence of different correlations among groups, a block diagonalization approach for the design of pre-beamformers achieves the optimal sum-rate scaling, albeit with a constant gap from the upper bound. Next, we consider the regime of large number of antennas and users, where user selection does not provide significant gain. In the presence of a large number of antennas, the design of prebeamformers reduces to choosing the columns of a Discrete Fourier Transform matrix based on the angles of arrival and angular spreads of the user channel covariance, when the base station (BS) is equipped with a uniform linear antenna array. Motivated by this result, we propose a simplified user grouping algorithm to cluster users into groups when the number of antennas becomes very large, in a realistic setting where users are randomly distributed and have different angles of arrival and angular spreads depending on the propagation environment. Our subsequent analysis leads to a probabilistic scheduling algorithm, where users within each group are preselected at random based on probabilities derived from the large system analysis, depending on the fairness criterion. This is advantageous since only the selected users are required to feedback their channel state information (CSIT).
\end{abstract}

{\bf Keywords:} JSDM, Opportunistic Beamforming, User Grouping, Probabilistic Scheduling.

\newpage

\section{Introduction}

Multiuser MIMO is one of the core technologies that has been adopted for the next generation of wireless communication systems. A considerable amount of effort has been dedicated to the study of such systems, where a transmitter (BS) equipped with multiple antennas serves a number of single antenna user terminals (UTs). Looking at the downlink scenario, we know that from an operational perspective, the high throughput promised by multiuser MIMO depends on the availability of accurate channel state information at the transmitter (CSIT). In Frequency division duplexing (FDD) systems, where uplink and downlink transmissions occur in separate bands, CSI is made available at the transmitter through downlink training and uplink feedback. In principle, the base station (BS) sends a sequence of orthogonal pilots which enables the users to estimate their own channels. These estimated channels are then fed back to the BS in a separate band. Analysis shows that for an appropriately designed feedback scheme, the channel estimation error due to feedback is negligible compared to the channel estimation error due to downlink training. The amount of training resources scales with the number of antennas at the BS.

In contrast, for a time division duplexing (TDD) system, channel reciprocity can be used to get estimates of the downlink channels through uplink training, thereby eliminating the need for feedback. Since the training dimension is now determined by the number of user terminals, the number of antennas can be made as large as desired. This approach, dubbed ``massive MIMO'' has garnered considerable interest because of simplified system operations in terms of scheduling and signal processing while maintaining the high performance gains promised by multiuser MIMO technology. Recently, Joint Spatial Division and Multiplexing (JSDM) was proposed to enable massive MIMO like gains and simplified operations for FDD systems, which represent the majority of currently deployed cellular networks. Making use of the fact that the channel between a UT and BS is correlated, the key idea lies in partitioning users into groups with similar covariance eigenspaces, and split the downlink beamforming into two stages: a pre-beamforming matrix that depends on the channel covariances, and a MU-MIMO precoding matrix for the effective channel formed by pre-beamforming. The pre-beamforming matrix is chosen in order to minimize the interference across different groups, and the MU-MIMO precoding matrix takes care of the interference within a group. The training dimensions required are for the design of the MU-MIMO precoder, which is significantly reduced after the pre-beamforming stage.

In our previous work  \cite{adhikary2012joint}, we have shown that under some conditions of the channel covariance eigenvectors, JSDM is optimal and achieves the capacity region. In the case when the BS is equipped with a uniform linear antenna array and the number of antennas is very large, the design of the pre-beamforming reduces to choosing certain columns of the Discrete Fourier Transform matrix based on the angles of arrival and angular spread of the user channel covariance. As long as the different user groups have non-overlapping supports of their angles of arrival and angular spreads, JSDM achieves optimality. This scheme is extremely beneficial in the sense that it requires only a coarse knowledge of the angular support, instead of the whole channel covariance matrix. The work in \cite{adhikary2012joint} assumed that users in a particular group had the same channel covariance structure. Furthermore, no user selection was considered, i.e., a certain number of users in each group was selected and scheduled for transmission
at random, such that their channel vectors preserve mutual statistical independence and some known technique based on large random matrix theory
for the performance analysis of the various pre-beamforming and precoding schemes can be applied in this context and provide an easy alternative to extensive system
simulation.  In this work, we consider two different set of results and regimes of operation.

\begin{itemize}
\item
First, we focus on a non-asymptotic regime in the number of base station antennas,
while we let the number of users in each group become large. In this context, we examine the performance of the well-known {\bf opportunistic beamforming}
scheme that serves on each downlink beam the user achieving the maximum SINR on that beam. It is well-known
that {\bf opportunistic beamforming with user selection does not provide any gain in the regime of large number of antennas} \cite{caire2009selection}.
In contrast, in the regime of fixed number of antennas and large number of users per group, we show that opportunistic beamforming yields significant gain,
and in fact channel correlation may yield a capacity improvement over the classical uncorrelated ``isotropic'' channel result of Sharif and Hassibi, because
of the fact that users come in groups, and in each group we can achieve both beamforming gain and multiuser spatial multiplexing.
More specifically,  in this regime of large number of users and fixed number of antennas,  the problem of sum capacity scaling with user selection has been widely
investigated for uncorrelated channels under random beamforming \cite{sharif2005capacity} and zero forcing beamforming \cite{yoo2006optimality},
and also for correlated channels under random beamforming \cite{al2009much}.  Our work differs from these earlier works in the sense that we consider
different correlations for different groups, which is an extension of \cite{al2009much} for multiple correlated channels.
We show that following a {\bf block diagonalization approach for the design of pre-beamformers achieves the optimal sum rate scaling},
albeit with a constant gap from the upper bound.

\item Then,  we focus again on the more interesting regime of a large number of antennas and users, where user selection becomes useless.
Differently from \cite{adhikary2012joint}, we consider the more realistic setting where the users are randomly distributed in the cellular region
and, as such, have different angles of arrival and angular spreads depending on the propagation environment.
We look at the {\bf problem of clustering users into groups} based on different user grouping algorithms, and evaluate their performance.
We show through finite dimensional simulations that choosing the pre-beamforming matrices as the columns of a discrete Fourier transform matrix
gives good results and, based on this observation, we propose a simplified user grouping algorithm when the number of antennas becomes very large (massive MIMO).
Motivated by the work of \cite{huh2012network}, \cite{huh2011achieving}, we focus on the regime where the number of users is proportional to the number
of antennas, and propose a probabilistic scheduling algorithm, where users within each group are pre selected at random based on probabilities
derived from the large system analysis and only the selected users are required
to feedback their CSIT.  Notice that in comparison with the regime of random beamforming and user selection considered before, in this regime
the CSIT feedback is limited because {\bf only the pre-selected (scheduled) users need to feed back their effective channels
(after pre-beamforming)}. In contrast, in the previous regime, CSIT feedback is limited by the fact that it is very simple (only
CQI and the beam index, as in \cite{sharif2005capacity}).

\end{itemize}

This report is organized as follows.
In Section \ref{sec:JSDM}, we briefly describe the channel model and the basic principles of the JSDM scheme.
We derive the sum capacity scaling result in Section \ref{sec:scaling}, by providing an upper and lower bounds
to the sum capacity in the presence of a large number of users.
The user grouping problem is addressed in Section \ref{sec:user-grouping}, and two algorithms are presented with their
performance evaluation through simulations. In Section \ref{sec:large-system-limit}, we focus on the large system limit,
when the number of users is proportional to the number of antennas and the number of antennas grows to infinity.
We derive a simplified user grouping algorithm that only requires the knowledge of the angles of arrival of the users,
and then propose our probabilistic scheduling algorithm along with some results.

{\em Notation} : We use boldface capital letters $\Xm$ for matrices, boldface small letters for vectors $\xv$, and small letters
$x$ for scalars. $\Xm^T$ and $\Xm^\herm$ denote the transpose and the Hermitian transpose of X, $||\xv||$ denotes the vector 2-norm of $\xv$, $\trace(\Xm)$ and $\det(\Xm)$ denote the trace and the determinant of the square matrix $\Xm$. The $n \times n$ identity matrix is denoted by $\Id_n$, and $||\Xm||^2_F = \trace(\Xm^\herm \Xm)$ indicates the
squared Frobenius norm of a matrix $\Xm$. We also use ${\rm Span}(\Xm)$ to denote the linear subspace generated by columns
of $\Xm$ and ${\rm Span}^\perp(\Xm)$ for the orthogonal complement of ${\rm Span}(\Xm)$. $\xv \sim \Cc \Nc(\muv,\Sigmam)$ indicates that $\xv$ is a complex circularly-symmetric Gaussian vector with mean $\muv$ and covariance matrix $\Sigmam$.

\section{Review of JSDM} \label{sec:JSDM}

In this section, we briefly describe the JSDM scheme proposed in \cite{adhikary2012joint}. The scheme relies on the fact that users are partitioned into different groups such that users within a group have approximately the same channel covariance structure and the different groups are have almost orthogonal covariances. The structure of the channel covariances is then exploited to form a reduced dimensional effective channel that enables the scheme to achieve large throughput gains with reduced training and feedback.

Consider the downlink of a cellular system formed by a BS having $M$ antennas and serving $K$ single antenna user terminals. We assume that the $M \times K$ dimensional channel matrix $\Hm$ is fixed for a certain block length of $T$ channel uses, which is known as the coherence time of the channel, and changes from block to block according to a ergodic stationary spatially white joint Gaussian process. A single channel use of such a system is denoted as
\begin{equation} \label{eqn:channel-use}
\yv = \Hm^\herm \xv + \zv = \begin{pmatrix} \hv_1^\herm \\ \hv_2^\herm \\ \vdots \\ \hv_K^\herm
\end{pmatrix} \Bm \Pm \dv + \zv
\end{equation}
where $\yv$ denotes the collection of received symbols for all the $K$ users, $\hv_k$ is the $M \times 1$ dimensional channel realization between the BS and UT $k$, $\xv = \Vm \dv$ is the transmitted signal vector satisfying a power constraint $P$ such that $\EE[||\xv||^2] \leq P$, $\Vm = \Bm \Pm$ is the downlink beamforming matrix consisting of two parts: $\Bm$ is the pre beamforming matrix of dimensions $M \times b$ and $\Pm$ is the multiuser MIMO precoding matrix of dimensions $b \times S$, which is a function of the reduced dimensional effective channel $\underline{\textsf{\Hm}} = \Bm^\herm \Hm$. $\dv$ is the $S \times 1$ vector of transmitted data streams. In general, we have $S \leq \min\{b,K\}$, and this represents the number of simultaneously served users per channel use. $\zv \sim \Cc \Nc (\zerov,\Id_K)$ is the corresponding additive white Gaussian noise vector whose entries and i.i.d. with zero mean and variance 1. $\hv_k$ is a correlated random vector with mean zero and covariance $\Rm_k$.

\subsection{Channel Model}

\begin{figure}[ht]
\centerline{\includegraphics[width=8cm]{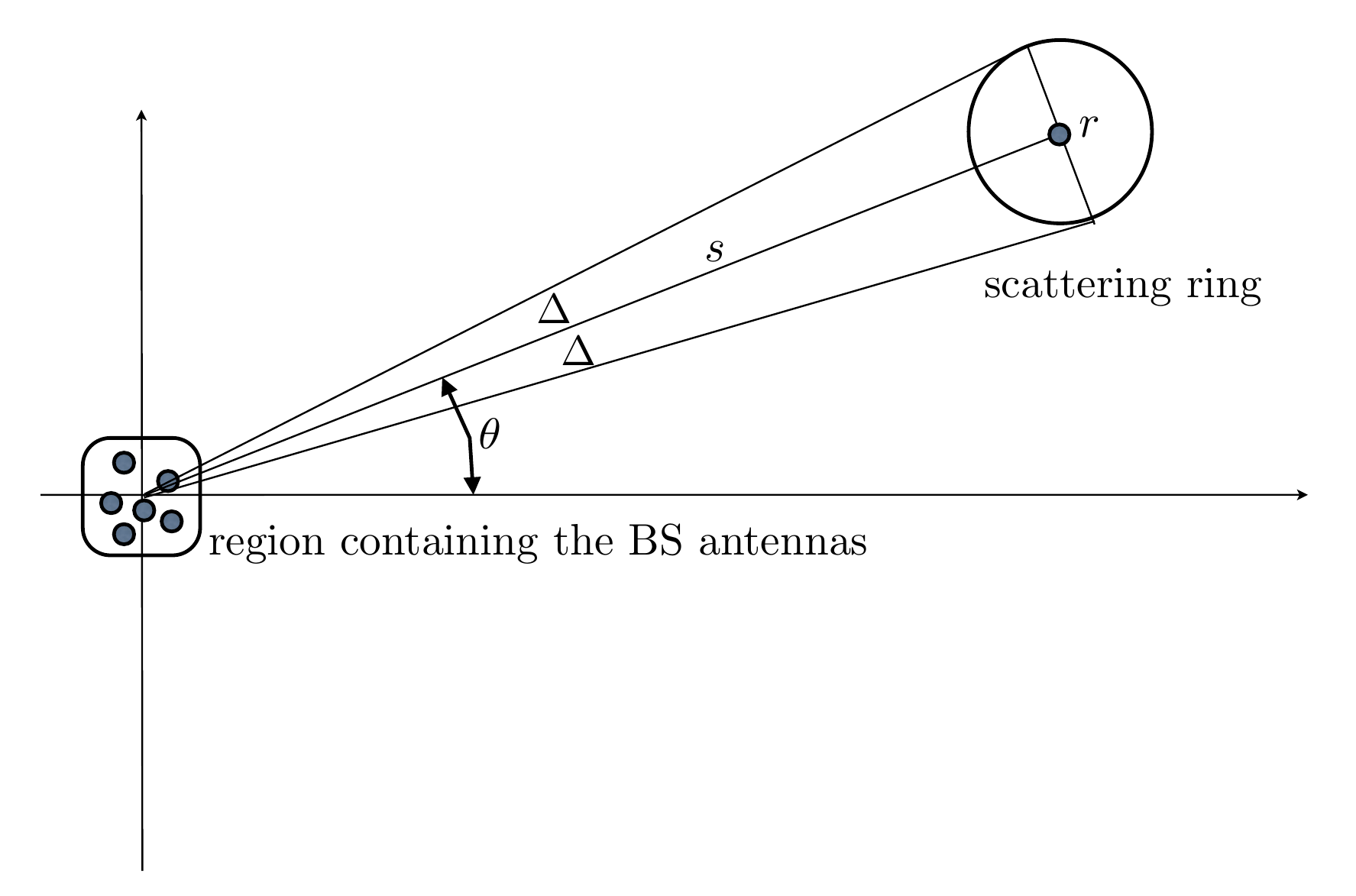}} \caption{A
UT at AoA $\theta$ with a scattering ring of radius ${\sf r}$
generating a two-sided AS $\Delta$ with respect to the BS at the
origin.} \label{fig:one-ring-model}
\end{figure}

In line with \cite{adhikary2012joint}, for analytical simplicity, we model the channel covariance $\Rm_k$ for a UT $k$ according to the one-ring model of
Figure \ref{fig:one-ring-model}, where a UT $k$ located at an azimuth angle $\theta$ and distance $\textsf{s}$ is surrounded by a ring of scatterers of radius $\textsf{r}$, giving the angular spread $\Delta = {\rm tan}^{-1} \left ( \frac{\textsf{r}}{\textsf{s}} \right )$. This model makes use of the fact that the BS antennas are located at the top of a tall building such that there is no significant scattering around the BS antennas. Assuming a uniform distribution of the received power of the planar waves impinging on the BS antennas, the entries of the channel covariance $\Rm_k$ are given by
\begin{equation} \label{eqn:channel-coeff}
[\Rm_k]_{m,p} = \frac{1}{2 \Delta} \int_{\theta - \Delta}^{\theta + \Delta} e^{j \kv^T(\alpha) (\uv_m - \uv_p)} d \alpha
\end{equation}
where $[\Rm_k]_{m,p}$ represents the channel correlation coefficient between the $m^{\rm th}$ and $p^{\rm th}$ transmit antennas of the BS, $\kv(\alpha) = -\frac{2 \pi}{\lambda} (\cos(\alpha),\sin(\alpha))^T$ is the wave vector for a planar wave with angle of arrival $\alpha$, $\lambda$ is the carrier wavelength and $\uv_m,\uv_p \in \RR^2$ are vectors indicating the position of the BS antennas in the two dimensional coordinate system.

Performing the Karhunen Loeve decomposition on $\Rm_k$, we have $\Rm_k = \Um_k \Lambdam_k \Um_k^\herm$, where $\Um_k$ is the $M \times r_k$ matrix containing the eigenvectors of the $r_k$ non-zero eigenvalues of $\Rm_k$ and $\Lambdam_k$ is the $r_k \times r_k$ matrix of non-zero eigenvalues. With this representation, we write the channel $h_k$ of user $k$ as
\begin{equation} \label{eqn:user-channel}
\hv_k = \Um_k \Lambdam_k^{\frac{1}{2}} \wv_k
\end{equation}
where the entries of $\wv_k$ are independent and identically distributed with mean zero and variance 1.

\subsection{The Basic Principle}

In JSDM, the $K$ UTs are partitioned into $G$ groups based on the similarity of their channel covariances. Denoting by $K_g$ and $S_g$ the number of UTs and the number of independent data streams in group $g$, we have $\sum_{g=1}^G K_g = K$ and $\sum_{g=1}^G S_g = S$, with the index $g_k$ used to denote the $k^{\rm th}$ user in group $g$. The channel vector of each user $g_k$ is given as $\hv_{g_k} = \Um_{g_k} \Lambdam_{g_k} \wv_{g_k}$ according to (\ref{eqn:user-channel}). Denoting by $\Hm_g = [\hv_{g_1} \ldots \hv_{K_g}]$ the concatenated channel matrix of UTs in group $g$, we have the overall $M \times K$ system channel matrix as $\Hm = [\Hm_1 \ldots \Hm_G]$. JSDM is a two stage precoding scheme, with the JSDM precoding matrix $\Vm = \Bm \Pm$ consisting of two parts: the pre-beamforming matrix $\Bm$ of dimensions $M \times b$, and the MU-MIMO precoder $\Pm$ of dimensions $b \times S$. The design of $\Bm$ is independent of the instantaneous channel realizations and is a function of the channel covariances of users in group $g$, i.e., it depends on the sets $\{\Um_{g_k},\Lambdam_{g_k}\}$. Alternately, $\Bm$ can be fixed apriori, for example, like the schemes of random beamforming \cite{sharif2005capacity}, Grassmannian beamforming \cite{love2003grassmannian}, etc. The multiuser MIMO precoding matrix $\Pm$ is dependent on the instantaneous ``effective'' channel $\underline{\textsf{\Hm}} = \Bm^\herm \Hm$. Denoting the pre-beamforming matrix of group $g$ as $\Bm_g$ of dimensions $M \times b_g$ such that $\sum_{g=1}^G b_g = b$, we have $\Bm = [\Bm_1, \ldots, \Bm_G]$. As a result, the received signal in (\ref{eqn:channel-use}) can be written in the following manner:
\begin{equation} \label{eqn:channel-use-jsdm}
\yv = \underline{\textsf{\Hm}}^\herm \Pm \dv + \zv
\end{equation}
where
\begin{equation}
\underline{\textsf{\Hm}}^\herm = \begin{pmatrix} \Hm_1^\herm \Bm_1 & \Hm_1^\herm \Bm_2 & \ldots & \Hm_1^\herm \Bm_G\\
\Hm_2^\herm \Bm_1 & \Hm_2^\herm \Bm_2 & \ldots & \Hm_2^\herm \Bm_G\\
\vdots & \vdots & \ddots & \vdots\\
\Hm_G^\herm \Bm_1 & \Hm_G^\herm \Bm_2 & \ldots & \Hm_G^\herm \Bm_G
\end{pmatrix}
\end{equation}
and $\Hm_g^\herm \Bm_{g'}$ denotes the effective channel matrix between the users of group $g$ and the pre-beamformers of group $g'$.

We focus of Per Group Processing (PGP), proposed in \cite{adhikary2012joint}, where the MU-MIMO precoding matrix $\Pm$ takes on the block diagonal form, i.e., $\Pm = {\rm diag}(\Pm_1,\ldots,\Pm_G)$ with $\Pm_g$ of dimensions $b_g \times S_g$. In other words, the MU-MIMO precoding matrix is designed independently across groups, meaning $\Pm_g$ is a function of the effective channels $\underline{\textsf{\Hm}}_g = \Hm_g^\herm \Bm_g$ only.
This approach is attractive since it requires only the knowledge of the effective channels
instead of the whole channel $\underline{\textsf{\Hm}}$. Focusing only on the received signal for users in group $g$, we have
\begin{equation}  \label{ziobanana-pgp}
\yv_g = \underline{\textsf{\Hm}}_g^\herm \Pm_g \dv_g + \left(\sum_{g'=1,g' \neq g}^G \Hm_g^\herm \Bm_{g'} \Pm_{g'} \dv_{g'} \right) + \zv_g
\end{equation}
where the bracketed term denotes the inter-group interference.

A suitable design goal for choosing $\Bm_g$ is to make the inter-group interference close to zero, meaning $\Hm_g^\herm \Bm_{g'} = \zerov \;\; \forall \; g' \neq g$. In \cite{adhikary2012joint}, assuming that users within a group have the same channel covariance and users across groups have different channel covariances, conditions for exact and approximate block diagonalization (BD) are obtained. For the purpose of approximate BD, in the event that exact BD is infeasible, the notion of approximate rank\footnote{Approximate rank is the number of dominant eigenvalues of the channel covariance.} is introduced, which is a design parameter that can be optimized. In the general system model considered here, where users in general have different channel covariances, the conditions for exact and approximate BD are given as follows:
\begin{itemize}
\item Exact BD : This is possible when ${\rm Span}(\Um_{g_1},\ldots,\Um_{g_{K_g}})$ has a non-empty intersection with ${\rm Span}^{\perp}(\Um_{g'_1},\ldots,\Um_{g'_{K_{g'}}} : g' \neq g)$. Since we are sending $S_g$ independent data streams to users in group $g$, this requires
    \begin{equation} \label{eqn:exact-BD}
    \dim\left({\rm Span}(\Um_{g_1},\ldots,\Um_{g_{K_g}}) \bigcap {\rm Span}^{\perp}(\Um_{g'_1},\ldots,\Um_{g'_{K_{g'}}} : g' \neq g) \right) \geq S_g
    \end{equation}
\item Approximate BD : Denoting by $\Um_{g_k}^*$ the set of eigenvectors corresponding to $r_k^*$ dominant eigenvalues of user $k$ in group $g$, we require ${\rm Span}(\Um^*_{g_1},\ldots,\Um^*_{g_{K_g}})$ to have non-empty intersection with ${\rm Span}^{\perp}(\Um^*_{g'_1},\ldots,\Um^*_{g'_{K_{g'}}} : g' \neq g)$. In order to be able to send $S_g$ independent data streams to users in group $g$, we need
    \begin{equation} \label{eqn:approx-BD}
    \dim\left({\rm Span}(\Um^*_{g_1},\ldots,\Um^*_{g_{K_g}}) \bigcap {\rm Span}^{\perp}(\Um^*_{g'_1},\ldots,\Um^*_{g'_{K_{g'}}} : g' \neq g) \right) \geq S_g
    \end{equation}
\end{itemize}
\begin{rem}
For finite $M$ and $K$, design methodologies for the pre-beamforming matrices to satisfy the conditions of exact and approximate BD are given in \cite{adhikary2012joint}. Furthermore, when $M$ is large, in the special case of uniform linear arrays, the channel covariance takes on a Toeplitz form. Owing to Szego's asymptotic theory \cite{grenander1958toeplitz}, \cite{adhikary2012joint}, the eigenvectors of the channel covariances can be well approximated by the columns of a Discrete Fourier Transform (DFT) matrix. In this special case, if users in different groups have disjoint angular support, their eigenvectors would be orthogonal (due to the property of the DFT matrix) and therefore, designing the pre-beamformers to attain exact BD is much simpler.
\end{rem}

\section{Sum Capacity Scaling for finite $M$ and large $K$} \label{sec:scaling}

In this section, we focus on the regime of finite $M$ and large $K$, and obtain
an asymptotic expression for the sum capacity when all the users within a group have the {\em same} channel covariance.
The case $G = 1$ is treated in \cite{al2009much}. Here, we consider the non-trivial extension to the case $G > 1$.
Denoting the covariance matrix of users in group $g$ as
$\Rm_g$ we have, by Karhunen Loeve decomposition
$$\Rm_g = \Um_g \Lambdam_g \Um_g^\herm,$$ where $\Um_g$ is the $M \times r_g$ matrix of
eigenvectors, $r_g$ is the rank of $\Rm_g$ and $\Lambdam_g$ is the diagonal matrix containing the eigenvalues of $\Rm_g$.
The channel of a user $k$ in group $g$ now takes the form $\hv_{g_k} = \Um_g \Lambdam_g^{1/2}
\wv_{g_k}$, $\wv_{g_k} \sim \Cc\Nc(\zerov,\Id_{r_g})$.
%


For the sake of mathematical simplicity, we assume that all groups contain the same number of users $K_g = K' = K/G$, for all $g$. We have:

\begin{thm}  \label{theorem1}
The sum capacity of a MU-MIMO downlink system with $M$ antennas,
total transmit power constraint of $P$, and $K$ users divided into $G$ groups of equal size $K' = K/G$,
where users  have mutually statistically independent channel vectors with common covariance matrix $\Rm_g$ to all
 users of each group $g$, behaves, for $K' \rightarrow \infty$, as
\begin{equation}
R_{\rm sum} \; = \; \beta
\log \log(K') + \beta  \log \left ( \frac{P}{\beta} \right)  + O(1)
\end{equation}
where $\beta = \min\{ M, \sum_{g=1}^G r_g \}$ and where $O(1)$ denotes a constant, independent of $K'$.
\end{thm}

Theorem \ref{theorem1} is proved by developing an upper and a lower bound. The upper bound analyzes directly the sum capacity of the underlying
vector broadcast channel, exploiting the sum capacity expression provided by the dual uplink channel \cite{vishwanath2003duality} (see Section \ref{converse-sec}).
Interestingly, in order to prove the lower bound we consider an explicit achievability strategy based on simple beamforming and user selection
in each group. This strategy generalizes the scheme of \cite{sharif2005capacity} (random beamforming) to the case where
the user are clustered in groups, each of which has a very strong directional component. As we shall see in Section \ref{subsec:achievability},
the achievability strategy consists of allocating  the user achieving the highest SINR on each beam of the pre-beamforming matrix, for each group.
Since the pre-beamforming matrices depend only on the channel second-order statistics, the feedback required from each user is just
the SINR achieved on each pre-beamfomring beam (or, equivalently, the max SINR and the index of the beam achieving this max SINR).
Hence, the achievability scheme has some practical interest since it is similar to the present ``opportunistic beamforming'' schemes
with Channel Quality Indicator (CQI) (see for example \cite{li2010mimo,ravindran2008multi}).

\subsection{Converse}  \label{converse-sec}

\setcounter{paragraph}{0}

\paragraph{Case $M > \sum_{g=1}^G r_g$}
Denoting the power allocated to a user $k$ in group $g$ as $P_{g_k}$,
letting $\Qm_g = {\rm diag}(P_{g_1},\ldots,P_{g_{(K')}})$ with trace $P_g = \sum_{k=1}^{K'} P_{g_k}$,
$\Hm_g = [\hv_{g_1} \ldots \hv_{g_{(K')}}]$ and owing to the uplink-downlink duality \cite{vishwanath2003duality},
we can write the sum capacity as
{\footnotesize{
\begin{eqnarray}
R_{\rm sum} &=&  \EE \left[ \max_{\sum_{g=1}^G \sum_{k=1}^{K'} P_{g_k} \leq P} \log \det \left( \Id_M +  \sum_{g=1}^G
\sum_{k=1}^{K'} \hv_{g_k} \hv_{g_k}^\herm P_{g_k} \right) \right]
\nonumber\\
&=& \EE \left[ \max_{\sum_{g=1}^G \sum_{k=1}^{K'} P_{g_k} \leq P} \log \det \left( \Id_M + \begin{pmatrix} \Hm_1 &
\ldots &
\Hm_G \end{pmatrix} \begin{pmatrix} \Qm_1 & \ldots & \zerov\\
\vdots & \vdots & \vdots\\
\zerov & \ldots & \Qm_G
\end{pmatrix} \begin{pmatrix} \Hm_1^\herm \\ \vdots \\
\Hm^\herm_G \end{pmatrix} \right) \right] \nonumber\\
&=& \EE \left[ \max_{\sum_{g=1}^G \sum_{k=1}^{K'} P_{g_k} \leq P} \log \det \left( \Id_{K} + \begin{pmatrix} \Hm_1^\herm \\ \vdots \\
\Hm^\herm_G \end{pmatrix} \begin{pmatrix} \Hm_1 &
\ldots & \Hm_G \end{pmatrix} \begin{pmatrix} \Qm_1 & \ldots & \zerov\\
\vdots & \vdots & \vdots\\
\zerov & \ldots & \Qm_G \end{pmatrix} \right) \right] \nonumber\\
&\substack{(a) \\ \leq}& \EE \left[\max_{\sum_{g=1}^G \sum_{k=1}^{K'} P_{g_k} \leq P} \log \det \left( \Id_{K} + \begin{pmatrix} \Hm_1^\herm \Hm_1 \Qm_1 & \ldots & \zerov\\
\vdots & \vdots & \vdots\\
\zerov & \ldots & \Hm_G^\herm \Hm_G \Qm_G \end{pmatrix} \right)
\right]\nonumber\\
&=& \EE \left[ \max_{\sum_{g=1}^G P_{g} \leq P} \left( \sum_{g=1}^G \log \det \left( \Id_{K'} +
\Hm_g^\herm \Hm_g \Qm_g \right) \right) \right]\nonumber\\
&=& \EE \left[ \max_{\sum_{g=1}^G P_{g} \leq P} \left( \sum_{g=1}^G \log \det \left( \Id_{M} + \Hm_g \Qm_g
\Hm_g^\herm \right) \right) \right] \nonumber\\
&=& \EE \left[ \max_{\sum_{g=1}^G P_{g} \leq P} \left( \sum_{g=1}^G \log \det \left( \Id_{M} + \Um_g
\Lambdam_g^{1/2} \Wm_g \Qm_g \Wm_g^\herm \Lambdam_g^{1/2}
\Um_g^\herm \right) \right) \right]
\nonumber\\
&=& \EE \left[ \max_{\sum_{g=1}^G P_{g} \leq P} \left( \sum_{g=1}^G \log \det \left( \Id_{r_g} +
\Lambdam_g^{1/2} \Wm_g \Qm_g \Wm_g^\herm \Lambdam_g^{1/2}
\Um_g^\herm \Um_g \right) \right) \right]
\nonumber\\
&=& \EE \left[ \max_{\sum_{g=1}^G P_{g} \leq P} \left( \sum_{g=1}^G \left[ \log \det (\Lambdam_g) \det \left(
\Lambdam_{g}^{-1} + \Wm_g \Qm_g \Wm_g^\herm \right) \right] \right) \right]
\nonumber\\
&\leq& \sum_{g=1}^G \log \det (\Lambdam_g) + \EE \left[ \max_{\sum_{g=1}^G P_{g} \leq P} \left( \sum_{g=1}^G r_g
\left[ \log \frac{{\rm tr}(\Lambdam_{g}^{-1} + \Wm_g \Qm_g
\Wm_g^\herm)}{r_g} \right] \right) \right] \nonumber\\
&\stackrel{(b)}{\leq}& \sum_{g=1}^G \log \det (\Lambdam_g) + \left[ \max_{\sum_{g=1}^G P_{g} \leq P} \left( \sum_{g=1}^G r_g
\log \left[ \frac{{\rm tr}(\Lambdam_{g}^{-1})}{r_g} + \EE \left [ \max_k \|\wv_{g_k}\|^2 \right ] \frac{P_g}{r_g} \right] \right) \right] \nonumber\\
&\stackrel{(c)}{=}& \sum_{g=1}^G \log \det (\Lambdam_g) + \max_{\sum_{g=1}^G P_{g} \leq P} \left[ \sum_{g=1}^G r_g \log \left[ \frac{{\rm tr}(\Lambdam_{g}^{-1})}{r_g} + \log(K') \frac{P_g}{r_g} + O(\log \log K') \right] \right]\nonumber\\
&=& \sum_{g=1}^G \log \det (\Lambdam_g) + \max_{\sum_{g=1}^G P_{g} \leq P} \left[ \sum_{g=1}^G r_g \log \left[ \log(K') \frac{P_g}{r_g} \left[\frac{{\rm tr}(\Lambdam_{g}^{-1})}{P_g \log K'} +  + O\left(\frac{\log \log K'}{\log K'}\right) \right] \right] \right]\nonumber\\
&=& \sum_{g=1}^G \log \det (\Lambdam_g) + \max_{\sum_{g=1}^G P_{g} \leq P} \left[ \sum_{g=1}^G r_g \log \left[ \log(K') \frac{P_g}{r_g} \right] + o(1) \right]\nonumber\\
\end{eqnarray}}}
where  $(a)$ is due to the Hadamard inequality for block matrices,
where $(b)$ follows from Jensen's inequality and $(c)$ follows from the fact that, for large $K'$,
\[ \EE[ \max_k \| \wv_{g_k} \|^2 ] = \log K'  + O(\log\log K') \]
(see Appendix \ref{subsec:extreme-val-theory}).  When $K \rightarrow \infty$, the upper bound can be further simplified as
\begin{equation}
R_{\rm sum} \leq \sum_{g=1}^G \log \det (\Lambdam_g) + \left (\sum_{g=1}^G r_g \right )
\log \log(K') + \max_{\sum_{g=1}^G P_{g} \leq P} \left[ \sum_{g=1}^G r_g \log \frac{P_g}{r_g} \right] + o(1)
\end{equation}
Optimizing the power allocation over groups, we obtain $P_g = \frac{r_g}{\sum_{g=1}^G r_g}P$, which gives
\begin{equation} \label{eqn:case1}
R_{\rm sum} \leq \sum_{g=1}^G \log \det (\Lambdam_g) + \left (\sum_{g=1}^G r_g \right )   \left [ \log \log(K') +
\log \frac{P}{\sum_{g=1}^G r_g} \right] + o(1)
\end{equation}

\paragraph{Case $M < \sum_{g=1}^G r_g$}
In this case, we write the sum capacity as
\begin{eqnarray}
R_{sum} &=&  \EE \left[ \max_{\sum_{g=1}^G \sum_{k=1}^{K'} P_{g_k} \leq P} \log \det \left( \Id_M +  \sum_{g=1}^G
\sum_{k=1}^{K'} \hv_{g_k} \hv_{g_k}^\herm P_{g_k} \right) \right]
\nonumber\\
&\leq& \EE \left[ \max_{\sum_{g=1}^G \sum_{k=1}^{K'} P_{g_k} \leq P} M \log \frac{\trace \left( \Id_M + \sum_{g=1}^G
\sum_{k=1}^{K'} \hv_{g_k} \hv_{g_k}^\herm P_{g_k} \right)}{M} \right] \nonumber\\
&=& \EE \left[ \max_{\sum_{g=1}^G \sum_{k=1}^{K'} P_{g_k} \leq P} M \log \left(1 + \sum_{g=1}^G
\sum_{k=1}^{K'} \frac{\trace \left( \hv_{g_k} \hv_{g_k}^\herm P_{g_k} \right)}{M} \right) \right] \nonumber\\
&=& \EE \left[ \max_{\sum_{g=1}^G \sum_{k=1}^{K'} P_{g_k} \leq P} M \log \left(1 + \sum_{g=1}^G
\sum_{k=1}^{K'} \frac{|| \hv_{g_k} ||^2 P_{g_k}}{M} \right) \right] \nonumber\\
&=& \EE \left[ \max_{\sum_{g=1}^G \sum_{k=1}^{K'} P_{g_k} \leq P} M \log \left(1 + \sum_{g=1}^G
\sum_{k=1}^{K'} \frac{\wv_{g_k}^\herm \Rm_g \wv_{g_k} P_{g_k}}{M} \right) \right] \nonumber\\
&\stackrel{(a)}{\leq}& \EE \left[ \max_{\sum_{g=1}^G \sum_{k=1}^{K'} P_{g_k} \leq P} M \log \left(1 + \sum_{g=1}^G
\sum_{k=1}^{K'} \frac{|| \wv_{g_k} ||^2 \lambda_{\max} P_{g_k}}{M} \right) \right] \nonumber\\
&\stackrel{(b)}{\leq}& \max_{\sum_{g=1}^G \sum_{k=1}^{K'} P_{g_k} \leq P} M \log \left(1 + \sum_{g=1}^G
\frac{\EE \left[ \max_k || \wv_{g_k} ||^2 \right] \lambda_{\max} \sum_{k=1}^{K'} P_{g_k}}{M} \right) \nonumber\\
&=& \max_{\sum_{g=1}^G \sum_{k=1}^{K'} P_{g_k} \leq P} M \log \left(1 + \sum_{g=1}^G
\frac{\lambda_{\max} P_g \log K' }{M} \right) + o(1) \nonumber\\
&=& M \log \left(1 + \frac{\lambda_{\max} P \log K' }{M} \right) + o(1) \nonumber\\
&=& M \log \lambda_{\max} + M \log \frac{P}{M} + M \log \log K' + o(1)
\end{eqnarray}
where (a) follows from the Rayleigh Ritz Theorem, for which $\wv_{g_k}^\herm \Rm_g \wv_{g_k} \leq \lambda_{\max,g} ||\wv_{g_k}||^2$,
where $\lambda_{\max,g}$ is the maximum eigenvalue of $\Rm_g$ and we let $\lambda_{\max} = \max_g \lambda_{\max,g}$. $(b)$ is due to Jensen's inequality.
Thus, we have established that
\begin{equation} \label{eqn:case2}
R_{sum} \leq M \log \lambda_{\max} + M \log \frac{P}{M} + M \log \log K'
\end{equation}

Combining (\ref{eqn:case1}) and (\ref{eqn:case2}), we can see that
\begin{equation}
R_{sum} \leq \beta \log \frac{P}{\beta} + \beta \log \log K' + O(1)
\end{equation}
where $\beta = \min \{M, \sum_{g=1}^G r_g\}$

\subsection{Achievability} \label{subsec:achievability}

We consider a specific JSDM strategy with PGP (see (\ref{ziobanana-pgp})) by letting the number of downlink data streams per group be given by
$b_g = S_g = r_g^*$ and the MU-MIMO precoding matrix in each group $g$ be the identity, i.e., $\Pm_g = \Id_{r_g^*} \;\;\; \forall \; g$.
In order to allocate the downlink data streams to the users, the scheme  selects $r_g^*$ out of $K'$ users in each group $g$ according to a max SINR
criterion to be specified later. Notice that since the achieved SINR for each user and pre-beamforming beam is a function of the
channel matrix realization,  this scheme serves $r^*_g$ out of $K'$ users  ``opportunistically'', depending on the channel matrix realization.
The pre-beamforming matrices $\Bm_g$ are designed according to the \emph{approximate Block Diagonalization} scheme,
where $r_g^*$ denotes the effective rank, as said before.
For any pair of groups $g, g'$, since the $m$-th column of $\Bm_{g'}$, denoted by $\bvx_{g'_m}$, is in the null space of the first $r_g^*$ eigenvectors
of $\Rm_g$ (dominant eigenvectors), we have that
\begin{equation}
\label{eqn:abd}
\Um_{g}^\herm \bvx_{g'_m} = \begin{pmatrix} \zerov_{r_g^* \times 1} \\
\xv_{g,g',m} \end{pmatrix}
\end{equation}
where $\xv_{g,g',m}$ is some not necessarily zero vector of dimension $r_g - r^*_g$.
Notice that when exact block diagonalization is possible and we choose $r_g^* = r_g$,
then $\Um_{g}^\herm \bvx_{g'_m} = \zerov_{r_g \times 1}$.

For the sake of convenience, let us focus on users in group $g$. We
assume that all the users $g_k$ have perfect knowledge of their SINR, with respect to beamforming vectors  $\bvx_{g_m}$ for $m = 1, \ldots, r_g^*$, given by
\begin{equation} \label{eqn:SINR-beam}
{\rm SINR}_{g_k,m} = \frac{|\hv_{g_k}^\herm \bvx_{g_m}|^2}{\frac{1}{\rho}
+ \sum_{n \neq m} |\hv_{g_k}^\herm \bvx_{g_n}|^2 + \sum_{g' \neq g}
||\hv_{g_k}^\herm \Bm_{g'} ||^2},
\end{equation}
where we let $\rho = \frac{P}{\sum_{g=1}^G r_g^*}$, assuming that the total transmit power is distributed evenly over all downlink beams.
Notice that such SINR is easily and accurately measured by including downlink pilot symbols
in the downlink streams passing through the pre-beamforming matrix, as currently done in opportunistic beamforming schemes \cite{viswanath2002opportunistic}, \cite{holma2007hsdpa}. Each user feeds back the
SINRs on all beams, i.e., for all $m  = 1,\ldots,r_g^*$, and
the BS decides to serve the user with the maximum SINR on a beam
$m$.\footnote{Following \cite{sharif2005capacity}, it is well-known that if each user feeds back just
its maximum SINR and the index of the beam achieving such maximum the achievable sum throughput
in the limit of large $K$ remains the same. Hence, instead of $r_g^*$ real numbers, the CQI feedback can be reduced to one real number and
an integer beam index. We omit this case since it follows trivially from previous work and does not change the final result.}
With this type of user selection, the achievable sum rate of group $g$ is given by
\begin{equation} \label{eqn:achieve-upper}
R_g = \sum_{m = 1}^{r_g^*} \EE \left[ \log \left( 1 + \max_{1 \leq k
\leq K'} {\rm SINR}_{g_k,m}\right) \right].
\end{equation}
Notice that with our assumptions it is possible that some user achieves the maximum on more than one beam, in which case the BS selects
to send multiple streams to that user.
\begin{rem}
It is proven in \cite{adhikary2012joint} that JSDM with PGP is optimal when the eigenvectors of the different groups satisfy the tall unitary condition.
When this is true, i.e., choosing $\Bm_g = \Um_g$ makes the inter-group interference term equal to zero ($\sum_{g' \neq g}
||\hv_{g_k}^\herm \Bm_{g'} ||^2 = 0$) and the numerator and denominator of the SINR term independent.
The analysis thus reduces to the approach of \cite{sharif2005capacity}, which gives the sum capacity scaling of $\log \log K'$ when $K' \rightarrow \infty$. When the tall unitary condition is not satisfied, choosing $\Bm_g = \Um_g$ gives a residual inter-group interference term and the numerator and denominator of the SINR term are no longer independent. It can be shown that in this case this simple user selection strategy does not achieve the $\log\log K'$ scaling.
When exact BD is possible, the inter group interference is zero and the problem decouples into an independent opportunistic beamforming scheme for each group, which
can be analyzed by direct application of the analysis technique of \cite{al2009much} in each group.
However, if exact BD  is not possible (e.g., when $\sum_g r_g > M$), we show in the following that by using approximate BD the scaling
$\log \log K'$ can also be achieved. In order to prove this result we combine the technique of \cite{al2009much} with some properties of
approximate BD.
\end{rem}


In order to find the scaling of the sum rate expressions (\ref{eqn:achieve-upper}) for large $K'$,
we consider the extremal statistics of ${\rm SINR}_{g_k,m}$, i.e., we study the distribution of the
random variable $\max_{1 \leq k \leq K'} {\rm SINR}_{g_k,m}$.
For this purpose, we find the distribution of a single term ${\rm SINR}_{g_k,m}$, whose CDF is given by
$F(x) = 1 - \PP\left ( {\rm SINR}_{g_k,m} > x \right )$. Define the quantity
\begin{eqnarray}
Z &=& |\wv_{g_k}^\herm \Lambdam_g^{1/2} \Um_g^\herm \bvx_{g_m}|^2 -
x \left[ \frac{1}{\rho} + \sum_{n \neq m} |\wv_{g_k}^\herm
\Lambdam_g^{1/2} \Um_g^\herm \bvx_{g_n}|^2 + \sum_{g' \neq g}
||\wv_{g_k}^\herm \Lambdam_g^{1/2} \Um_g^\herm \Bm_{g'} ||^2 \right]
\nonumber\\
&=& \wv_{g_k}^\herm \Am_{m_1} \wv_{g_k} - x \wv_{g_k}^\herm
\Am_{m_2} \wv_{g_k}^\herm - \frac{x}{\rho}
\end{eqnarray}
Following the analysis of \cite{sharif2005capacity} (see Appendix \ref{subsec:cdf}), we get
\begin{equation}
\label{eqn:pdf-sinr}
\PP \left ( {\rm SINR}_{k,m} > x \right ) =
\frac{1}{2\pi j} \int_{-\infty}^{\infty} \frac{e^{-(j\omega +
c)x/\rho}}{j\omega + c} \frac{1}{\prod_{i=1}^{r_g} (1 -
(j\omega + c)\mu_{m,i}(x))} d\omega
\nonumber\\
\end{equation}
where $\{\mu_{m,i}(x) : i  = 1,\ldots,r_g\}$ are the eigenvalues of $\Am_m(x) = \Am_{m_1} - x\Am_{m_2}$.

In order to derive the CDF of the SINR, we need to make some remarks
on the eigenvalues of $\Am_m(x)$. Ordering the eigenvalues of $\Am_m(x)$ as
\[ \mu_{m,1}(x) \geq \mu_{m,2}(x) \geq \ldots \mu_{m,r_g}(x),\]
we have the following lemmas
\begin{lem}
The maximum eigenvalue of $\Am_m(x)$, i.e., $\mu_{m,1}(x)$ is strictly
positive $\forall \ \ x \geq 0$.
\end{lem}

\begin{proof}
Rewriting $\Am_m(x)$ as
\begin{eqnarray}
\Am_m(x) &=& \Am_{m_1} - x\Am_{m_2} \nonumber\\
&=& \Lambdam_g^{1/2} \Um_g^\herm \bvx_{g_m} \bvx_{g_m}^\herm
\Um_g \Lambdam_g^{1/2}   - x \left[ \sum_{n \neq m} \Lambdam_g^{1/2}
\Um_g^\herm \bvx_{g_n} \bvx_{g_n}^\herm \Um_g \Lambdam_g^{1/2}
 + \sum_{g' \neq g} \Lambdam_g^{1/2} \Um_g^\herm \Bm_{g'}
\Bm_{g'}^\herm \Um_g \Lambdam_g^{1/2}
\right]
\end{eqnarray}
we see that $\Am_{m_1}$ is a rank-1 matrix and $\Am_{m_2}$ has rank at most $r_g - 1$. This is because $\Am_{m_2}$ is the sum of $r_g^* - 1$ rank-1 matrices and the matrix $\sum_{g' \neq g} \Lambdam_g^{1/2} \Um_g^\herm \Bm_{g'}
\Bm_{g'}^\herm \Um_g \Lambdam_g^{1/2}$ has at most rank $r_g - r_g^*$ because of (\ref{eqn:abd}). Since $\Am_{m_2}$ is of dimension $r_g \times r_g$ and has rank at most $r_g - 1$, there is a non-trivial nullspace of dimension 1, meaning we can find a vector $\qv$ such that $\Am_{m_2} \qv = \zerov$. In order to prove the lemma,
we first prove that
\begin{equation}\label{eqn:span}
\Um_g^\herm \bvx_{g_m} \notin {\rm Span} \left \{ \Um_g^\herm \bvx_{g_n} : n \neq m, \Um_g^\herm \bvx_{g'_n} : g' \neq g, n = 1,\ldots,r_{g'}^{*} \right \}.
\end{equation}
In order to see this, we write $\Um_g = [\Um_g^*, \Um'_g]$, where $\Um_g^*$ is of rank $r_g^*$ and $\Um'_g$ is of rank $r_g - r_g^*$. Let $\Bm_g = [\bvx_{g_1} \bvx_{g_2} \ldots \bvx_{g_{r_g^*}}]$, of rank $r_g^*$ by construction. Then, we have
\[ \Um_g^\herm \Bm_g = \begin{pmatrix} \Um_g^{* \herm} \Bm_g \\ \Um^{' \herm}_g \Bm_g \end{pmatrix}, \]
where the upper part $\Um_g^{* \herm} \Bm_g$ has rank $r_g^*$. Reasoning by contradiction, let's assume that (\ref{eqn:span})  is false.
Then, there exist coefficients $\{\alpha_{g_n}: n \neq m\}$ and $\{\beta_{g'_n} : g' \neq g, n=1,\ldots, r^*_{g'}\}$ such that
\[ \Um_g^\herm \bvx_{g_m} = \sum_{n\neq m} \alpha_{g_n} \Um_g^\herm \bvx_{g_n} + \sum_{g' \neq g} \sum_{n=1}^{r^*_{g'}} \beta_{g'_n} \Um_g^\herm \bvx_{g'_n}.  \]
Recalling  (\ref{eqn:abd}), we have that the second term in the right-hand side of the above equality takes on the form
\[   \sum_{g' \neq g} \sum_{n=1}^{r^*_{g'}} \beta_{g'_n} \Um_g^\herm \bvx_{g'_n} =  \begin{pmatrix} \zerov_{r_g^* \times 1} \\
\zv \end{pmatrix}, \]
where $\zv$ is some non-zero vector of dimension $r_g - r^*_g$.
Since the upper part of $\Um_g^\herm \bvx_{g_m}$, formed by the first $r^*_g$ components  $\Um_g^{* \herm} \bvx_{g_m}$, is non-zero. It must be
\[ \Um_g^{* \herm} \bvx_{g_m} = \sum_{n\neq m} \alpha_{g_n} \Um_g^{* \herm} \bvx_{g_n}.  \]
However, this cannot be, since it contradicts the fact that $\Um_g^{* \herm} \Bm_g$ has rank $r_g^*$.
Therefore, we conclude that (\ref{eqn:span}) holds.

Now, choosing $\qv$ to be a unit vector in the orthogonal complement of ${\rm Span}\{ \Um_g^\herm \bvx_{g_n} : n \neq m, \Um_g^\herm \bvx_{g'_n} : g' \neq g,
n =1,\ldots,r_{g'}^{*}\}$ and such that $\qv^\herm \Um_g^\herm \bvx_{g_m} \neq \zerov$, we have that
\[ 0 < \qv^\herm \Um_g^\herm \bvx_{g_m} \bvx_{g_m}^\herm \Um_g \qv \leq \max_{\qv} \qv^\herm \Um_g^\herm \bvx_{g_m} \bvx_{g_m}^\herm \Um_g \qv \eqdef \mu_{m,1}(x), \]
implying $\mu_{m,1}(x) > 0$ for all $x \geq 0$.
\end{proof}

\begin{lem}
The eigenvalues
$\mu_{m,2}(x),\ldots,\mu_{m,r_g}(x)$ are non-positive $\forall \ \ x \geq 0$.
\end{lem}

\begin{proof}
Denoting by $\lambda_{i}(\Am_{m_1})$ and $\lambda_{i}(\Am_{m_2})$ the $i^{\rm th}$ largest eigenvalues of $\Am_{m_1}$ and $\Am_{m_2}$, we have for $i > 1$, using Weyl's inequality \cite{franklin2012matrix}, we have
\begin{eqnarray}
\mu_{m,i}(x) &\leq& \lambda_{i}(\Am_{m_1}) - x \lambda_{r_g}(\Am_{m_2})\nonumber\\
&\leq& 0 - x \lambda_{r_g}(\Am_{m_2}) \nonumber\\
&\leq& 0
\end{eqnarray}
implying $\mu_{m,i}(x) \leq 0$, $\forall \; i > 1$
\end{proof}

Since the eigenvalues $\mu_{m,2},\ldots,\mu_{m,r_g}$ are negative and do
not contribute to the integral (\ref{eqn:pdf-sinr}), we can use Cauchy's integral theorem and the fact that there is a single pole
in the right half-plane of the complex plane in order to obtain
\begin{equation}
\PP \left ( {\rm SINR}_{g_k,m} > x \right ) = \frac{e^{\frac{-x}{\rho
\mu_{m,1}(x)}}}{\prod_{i=2}^{r_g} \left( 1 -
\frac{\mu_{m,i}(x)}{\mu_{m,1}(x)}\right)},
\end{equation}
such that the SINR CDF $F(x)$ is given by
\begin{equation} \label{eqn-cdf-sinr}
F(x) = 1 - \frac{e^{\frac{-x}{\rho \mu_{m,1}(x)}}}{\prod_{i=2}^{r_g} \left( 1 -
\frac{\mu_{m,i}(x)}{\mu_{m,r_g}(x)}\right)}.
\end{equation}

From the well known results on extreme value theory (see Appendix \ref{subsec:extreme-val-theory},\cite{sharif2005capacity},\cite{al2009much}), we have that $\max_{1 \leq k \leq K'} {\rm SINR}_{g_k,m}$ for a
group $g$ behaves as $\rho \mu_{m,1}^* \log K' + O(\log \log K')$ as $K' \rightarrow \infty$, where
\begin{equation} \label{eqn:limit-growth}
\rho \mu_{m,1}^* = \lim_{x \rightarrow \infty} g(x),
\end{equation}
and $g(x)$ denotes the growth function of the CDF $F(x)$ (see Appendix \ref{subsec:compute-limit-growth}).

As a result, the sum rate for a group $g$ behaves as
\begin{eqnarray*}
R_g &=& \sum_{m=1}^{r_g^*} \log \left( \rho \mu_{m,1}^* \log(K') \right) + o(1)\\
&=& r_g^* \log \rho + r_g^* \log \log K' + \sum_{m=1}^{r_g^*} \log \left( \mu_{m,1}^*
\right) + o(1)
\end{eqnarray*}
as  $K' \rightarrow \infty$, where $\rho = \frac{P}{\sum_{g=1}^G r_g^*}$,
assuming that each downlink stream is allocated equal power.
Following a similar approach for all the groups, we arrive at the sum rate achievable asymptotic formula
\begin{equation}
R_{\rm sum} = \left(\sum_{g=1}^{G} r_g^*\right) \log \rho + \left( \sum_{g=1}^G r_g^* \right) \log \log K' + O(1).
\end{equation}
When $\sum_{g=1}^G r_g < M$, it is possible to choose $r_g^* = r_g$ such that the above achievable sum rate matches
(in the leading terms) the upper bound (\ref{eqn:case1}).   If $\sum_{g=1}^G r_g > M$, we can choose
$r_g^*$ such that $\sum_{g=1}^G r_g^* = M$, such that again the achievable rate matches,
in the leading terms, the upper bound (\ref{eqn:case2}). Hence, in all cases, we have
$\sum_{g=1}^G r_g^* = \min \{M,\sum_{g=1}^G r_g\} = \beta$, such that
Theorem \ref{theorem1} is proved.

\section{User Grouping} \label{sec:user-grouping}

As a matter of fact, in reality users do not come naturally partitioned in groups with the same covariance matrix.
In order to exploit effectively the JSDM approach, the system must partition the users' population into groups according to the following qualitative principles:
1) users in the same group have channel covariance eigenspace spanning (approximately) a given common subspace, which characterizes the group;
2) the subspaces of groups served on the same time-frequency slot (transmission resource) by JSDM must be (approximately) mutually orthogonal, or at least have
empty intersection.  In this section, we focus on this user grouping problem when the BS is equipped with a uniform linear array.
According to the single scattering ring model (see \cite{adhikary2012joint} and references therein),
the channel covariance matrix for a user located at angle of arrival (AoA) $\theta$ with an angular spread (AS) $\Delta$ has $(m,p)$-th elements
\begin{equation}  \label{correlation-matrix-ULA}
[\Rm]_{m,p} = \frac{1}{2\Delta} \int_{-\Delta+\theta}^{\Delta+\theta}   e^{-j2\pi D (m-p) \sin(\alpha)} d\alpha,
\end{equation}
where $\lambda D$ denotes the minimum distance between the BS antenna elements.
Here we assume that users are characterized by the pair $(\theta, \Delta)$ of their AoA and AS,
depending on their location relative to the BS antenna array, and on this local scattering environment.
We consider  two user grouping algorithms and demonstrate their performance by simulation. In Section \ref{sec:large-system-limit},
we focus on the most promising scheme and consider the system performance analysis in the large system regime \cite{adhikary2012joint}, i.e., when both $K$ and $M$ are large.
We assume that the BS has perfect knowledge of the user channel covariance, which can be accurately learned and tracked since it is constant
in time.\footnote{As a matter of fact, the channel covariance is slowly varying in time, depending on the user mobility. However, especially for nomadic users,
the scattering environment characteristics and the AoA evolve in time much more slowly than the actual channel fading process, and can be considered
``locally constant''. Algorithms for covariance estimation and signal subspace tracking are well known and widely investigated, and
are out of the scope of this work.}

\subsection{Algorithm 1: K-means Clustering}

$K$-means Clustering is a standard iterative algorithm which aims at partitioning $K$ observations into $G$ clusters such that each
observation belongs to the cluster with the nearest mean \cite{young1974classification,lloyd1982least}.
This results in a partition of the observation space into Voronoi cells. In our problem, the $K$ user covariance {\em dominant}
eigenspaces, i.e., $\{\Um_k^* : k = 1,\ldots,K \}$ form the observation space.
Hence, in order to apply the K-means principle, we consider the {\em chordal distance} between the covariance enigenspaces.\footnote{Note that K-means is
usually formulated in terms of the Euclidean distance since  the observation space is typically a subset of $\CC^n$.}
Given two matrices $\Xm \in \CC^{M \times p}$ and $\Ym \in \CC^{M \times q}$, the chordal distance denoted by
$d_C(\Xm,\Ym)$ is defined by
\begin{equation} \label{eqn:chordal-dist}
d_C(\Xm,\Ym) = \left \| \Xm \Xm^\herm - \Ym \Ym^\herm \right \|_F^2.
\end{equation}
In a similar fashion, we need to define a notion of the {\em mean} of (tall) unitary matrices.  Given $N$ unitary matrices $\{\Um_1^*,\Um_2^*,\ldots,\Um_N^*\}$,
the mean $\bar{\Um}^* \in \CC^{M \times p}$ is given as \cite{barg2002bounds}
\begin{equation} \label{eqn:matrix-mean}
\bar{\Um}^* = {\rm eig} \left[ \frac{1}{N} \sum_{n=1}^N \Um_n^* \Um_n^{* \herm} \right],
\end{equation}
where ${\rm eig} (\Xm)$ denotes the unitary matrix formed by the $p$ dominant eigenvectors  of $\Xm$.

At this point, we can formulate the K-means algorithm for the user channel eigenspaces.
Given $K$ user covariance eigenspaces $\{\Um_k^* \in \RR^{M \times \bar{r}_k^*} : k = 1,\ldots,K \}$,
we need to cluster them into $G$ groups, where each group $g$ is characterized by its subspace (tall unitary matrix) $\Vm_g  \in \RR^{M \times r_g^*}$, such that
$\sum_{g=1}^G r_g^* \leq M$. We denote by $\Vm_g^{* (n)}$ the group $g$ ``mean'' obtained by the algorithm at iteration $n$,
and by $\Sc_g^{(n)}$ the set of users belonging to group $g$ at iteration $n$. We have:
\begin{itemize}
\item \textbf{Step 1}: Set $n = 0$ and $\Sc_g^{(0)} = \emptyset$ for $g = 1,\ldots,G$.
Randomly choose $G$ different indices from the set $\{1,\ldots,K\}$ and do the following assignment
\begin{equation}
\Vm_g^{* (n)} = \Um_{\pi(g)}^*, \;\;\; \mbox{for} \; g = 1, 2, \ldots, G,
\end{equation}
where $\pi(g)$ returns a random number from the set $\{1,2,\ldots,K\} \setminus \{\pi(1),\ldots,\pi(g-1)\}$
\item \textbf{Step 2}: For $k = 1,\ldots,K$, compute
\begin{equation}
d_C(\Um_k^*,\Vm_g^{* (n)}) = ||\Um_k^* \Um_k^{* \herm} - \Vm_g^{* (n)} \Vm_g^{* (n) \herm} ||_F^2
\end{equation}
\item \textbf{Step 3}: Assign user $k$ to group $g$ such that
\begin{eqnarray}
g &=& {\rm arg} \min_{g'} d_C(\Um_k^*,\Vm_{g'}^{* (n)}) \nonumber\\
\Sc_g^{(n+1)} &=& \Sc_g^{(n)} \cup \{k\}
\end{eqnarray}
\item \textbf{Step 4}: For $g = 1, \ldots, G$ and $\forall\ k \in \Sc_g$ compute
\begin{equation}
\Vm_g^{* (n+1)} = {\rm eig} \left[ \frac{1}{|\Sc_g^{(n+1)}|} \sum_{k \in \Sc_g^{(n+1)}}^N \Um_k^* \Um_k^{* \herm} \right]
\end{equation}
\item \textbf{Step 5}: Compute the total distance at the $n^{\rm th}$ and $(n+1)^{\rm th}$ iteration
\begin{equation}
d_{{\rm tot},C}^{(n)} = \sum_{g=1}^G \sum_{k \in \Sc_g^{(n)}} d_C(\Um_k^*,\Vm_g^{* (n)})
\end{equation}
\item \textbf{Step 6}: If $|d_{{\rm tot},C}^{(n)} - d_{{\rm tot},C}^{(n+1)}| > \epsilon d_{{\rm tot},C}^{(n)} $, go to Step 7. Else, increment $n$ by 1
and go to Step 2.\footnote{$\epsilon$ is a threshold for stopping the algorithm when the relative difference between the total distances at the previous and current iterations
is sufficiently small.}
\item \textbf{Step 7}: For $g =1,\ldots,G$, assign
$$\Vm_g^{*} = \Vm_g^{* (n)}, \ \ \Sc_g = \Sc_g^{(n)}.$$
\end{itemize}

\subsection{Algorithm 2: fixed quantization} \label{subsec:fixed-quantz}

In this case, the group subspaces $\{\Vm_g^* \in \RR^{M \times r_g^*} : g  = 1,\ldots,G\}$ are fixed and given {\em a priori}, based on
geometric considerations. They act as the representative points of a minimum distance quantizer, where
distance in this case is the {\em chordal distance} defined in (\ref{eqn:chordal-dist}). Explicitly, we have:

\begin{itemize}
\item \textbf{Step 1}: For $g =1,\ldots,G$ set $\Sc_g = \emptyset$.
\item \textbf{Step 2}: For $k = 1,\ldots,K$, compute the distances
\begin{equation}
d_C(\Um_k^*,\Vm_g) = ||\Um_k^* \Um_k^{* \herm} - \Vm_g^* \Vm_g^{* \herm} ||_F^2,
\end{equation}
find in the minimum distance group index
\[ g = {\rm arg} \min_{g'} d_C(\Um_k^*,\Vm_{g'}^*), \]
and add user $k$ to group $g$, i.e.,  let $\Sc_g := \Sc_g \cup \{k\}$.
\end{itemize}

It is clear that the performance of the JSDM scheme resulting from fixed quantization depends critically on how we choose the group subspaces.
We have considered two methods to choose $\{\Vm^*_g\}$.
The first method relies on the fact that, for large $M$, the channel eigespaces are nearly mutually orthogonal
when the channel AoA supports are disjoint \cite{adhikary2012joint}.  Hence, we choose
$G$ AoAs $\theta_g$ and fixed AS $\Delta$ such that the resulting $G$ intervals  $[\theta_g - \Delta, \theta_g + \Delta]$ are disjoint,
and compute the eigenspace corresponding to these artificially constructed covariance matrices using the one-ring scattering model
(\ref{correlation-matrix-ULA}).  This method consists essentially to form pre-defined ``narrow sectors'' and associate
users to sectors according to minimum chordal distance quantization.

\begin{example} \label{ex:method-1}
Suppose $G = 3$. Choosing $\theta_1 = -45^o, \theta_2 = 0^o, \theta_3 = 45^o$ and $\Delta = 15^o$, we note that the angular supports are disjoint.
Letting $\Rm_1(\theta_1,\Delta), \Rm_2(\theta_2,\Delta)$ and $\Rm_3(\theta_3,\Delta)$ denote the covariance matrices
obtained by (\ref{correlation-matrix-ULA}) for given AoA and AS,  we let $\Vm_g^* = \Um_g^*$ for $g  = 1,2,3$, where $\Um_g^*$
is the $M \times r_g^*$ tall unitary matrix of the $r_g^*$ dominant eigenvalues of $\Rm_g(\theta_g,\Delta)$ and
the effective ranks $r_g^*$ are chosen such that $r_1^* + r_2^* + r_3^* = M$.
\end{example}

A different way to choose $\Vm_g^*$ consists of maximizing the minimum distance between the group subspaces.
Defining
\begin{equation}
d_{\Vm_g^*: g \in \{1,2,\ldots,G\}} =  \min_{g,g'} d_C(\Vm_g^*,\Vm_{g'}^*)
\end{equation}
as the minimum chordal distance of the set of group subspaces $\{\Vm_1^*,\Vm_2^*,\ldots,\Vm_G^*\}$, we wish to find such set
such that $d_{\Vm_g^*: g \in \{1,2,\ldots,G\}}$ is maximized.  It is easy to see that, if  $\sum_{g=1}^G r_g^* = M$, the we can choose $\{\Vm_g\}$ as disjoint
subsets of the columns of a unitary matrix of dimensions $M \times M$ such that all group subspaces are mutually orthogonal and
$d_{\Vm_g^*: g \in \{1,2,\ldots,G\}}$ is maximized.
Using the fact that, for large $M$,  the eigenvectors of covariance matrices of the type (\ref{correlation-matrix-ULA})
are well approximated by the columns of a DFT matrix (see \cite{adhikary2012joint} for details), here we propose to use disjoint blocks of
adjacent columns of the $M\times M$ unitary DFT matrix as group subspaces.

\begin{example} \label{ex:method-2}
Suppose again $G = 3$. Assign $r_g^* = \lfloor \frac{M}{3} \rfloor = r$ and let $\Fm$ denote the unitary $M \times M$ DFT matrix. Then, we have
$$\Vm_g = \Fm(:,(g-1)r + (1:r)),$$
i.e., $\Vm_g$ is formed by taking the $(g-1)r + 1$ to $(g-1)r + r$ columns of $\Fm$.
\end{example}

\subsection{Simulations} \label{sec:sim-finite}

We present some simulation results to show the performance of the different user grouping algorithms proposed in Section \ref{sec:user-grouping}, under
the pre-beamforming and user selection JSDM scheme used in the achievability of Theorem \ref{theorem1} (see Section \ref{subsec:achievability}).
Having chosen the group eigenspaces $\{\Vm_g^*\}$ and the set of users $\Sc_g$ for each group $g$, we obtain the pre-beamforming matrices
$\{\Bm_g\}$ by block diagonalization (see \cite{adhikary2012joint} for details). In particular, we have
$(\Vm^*_{g})^\herm \Bm_g = \zerov$ for all $g \neq g'$.  Notice that this does not mean that the system has no inter-group interference, since the
actual user channels eigenspaces do not coincide {\em exactly} with the group subspaces.
Within a group, we adopt user selection to choose a subset $r_g^*$ of users among $K_g = |\Sc_g|$ to be served during each particular time-frequency slot.
The user selection algorithm  used in the achievability of Theorem \ref{theorem1} is denoted here as {\em JSDM-GBF-ALL} (group beamforming with all SINRs CQI feedback).
For the sake of comparison we consider other two selection schemes. The first, denoted as {\em JSDM-GBF-MAX}, consists of feeding back
just the index of the beam with max SINR. Such scheme was considered and analyzed in \cite{sharif2005capacity} for i.i.d. channel vectors.
The second, denoted as {\em JSDM-ZFBF-SUS}, consists of performing ZFBF precoding for each group,
where users in each group $g$ are chosen by semi-orthogonal user selection (SUS) proposed in \cite{yoo2006optimality},
on the basis of the effective channel vectors including pre-beamforming, $\{\Bm_g^\herm \hv_k : k \in \Sc_g\}$.
The selected users are served by ZFBF MU-MIMO precoding, where the precoding matrix
$\Pm_g$ is the column-normalized Moore-Penrose pseudo-inverse of effective channel matrix formed by  selected users.
Notice that  {\em JSDM-ZFBF-SUS} requires feeding back the effective channel vectors and therefore incurs in a much larger feedback overhead.

We consider two specific cases: $M = 8$ and $M = 16$.
We fix the total transmit power $P =\ 10 dB$ assuming equal power per stream and normalize the noise variance to 1,
and denote ${\rm SNR} = P$ in the plots.
We set $G = 8$. The angles of arrival for the users are generated randomly between $-60^o$ to $60^o$ and the angular spreads are generated
randomly between $5^o$ to $15^o$.
For the $K-$means clustering algorithm, the entire set of user covariances is clustered into $G = 8$ groups. For the fixed quantization algorithm,
we choose $\theta \in \{ -57.5^o, -41.5^o, -23^o, -7.5^o, 7.5^o, 23. 5^o, 41.5, 57.5^o\}$ and $\Delta = 12^o$ for choosing our group subspaces, as shown in Example \ref{ex:method-1}.  For the DFT based fixed quantization scheme as in Example \ref{ex:method-2},
we choose $\Vm_g^* = \Fm(:,{\rm Mod}_M[(g-1)r + (1:2r)])$, where the ${\rm Mod}_M$ operation ensures that the column indices are in the set $\{1,2,\ldots,M\}$, and we use
$r = 1$ for $M = 8$ and $r = 2$ for $M = 16$.
Once the clustering is done, we further separate the groups into two disjoint subsets, referred to as ``patterns'',
each containing $G/2 = 4$ groups. Groups in the same pattern are served on the same time-frequency slot,
while users in different pattern are served in different slots. This partitioning into patterns is needed in order to keep the inter-group interference under control.
For the $K-$means clustering algorithm, the partitioning of the group eigenspaces into two disjoint patterns is done such that the sum of the minimum distances
of the two patterns is maximized. For the fixed quantization algorithm, the patterns are obtained by considering the geometry of angular separation. In particular,
we have the two patterns: $\{\Vm_1^*,\Vm_3^*,\Vm_5^*,\Vm_7^*\}$ and $\{\Vm_2^*,\Vm_4^*,\Vm_6^*,\Vm_8^*\}$.
Within each group, a specific user selection algorithm ({\em JSDM-ZFBF-SUS}, {\em JSDM-GBF-MAX} or {\em JSDM-GBF-ALL}) is applied for
selecting the users which are served at any given time-frequency slot.

\begin{figure}
\centering \subfigure[$M = 8$]{
  \includegraphics[width=8cm]{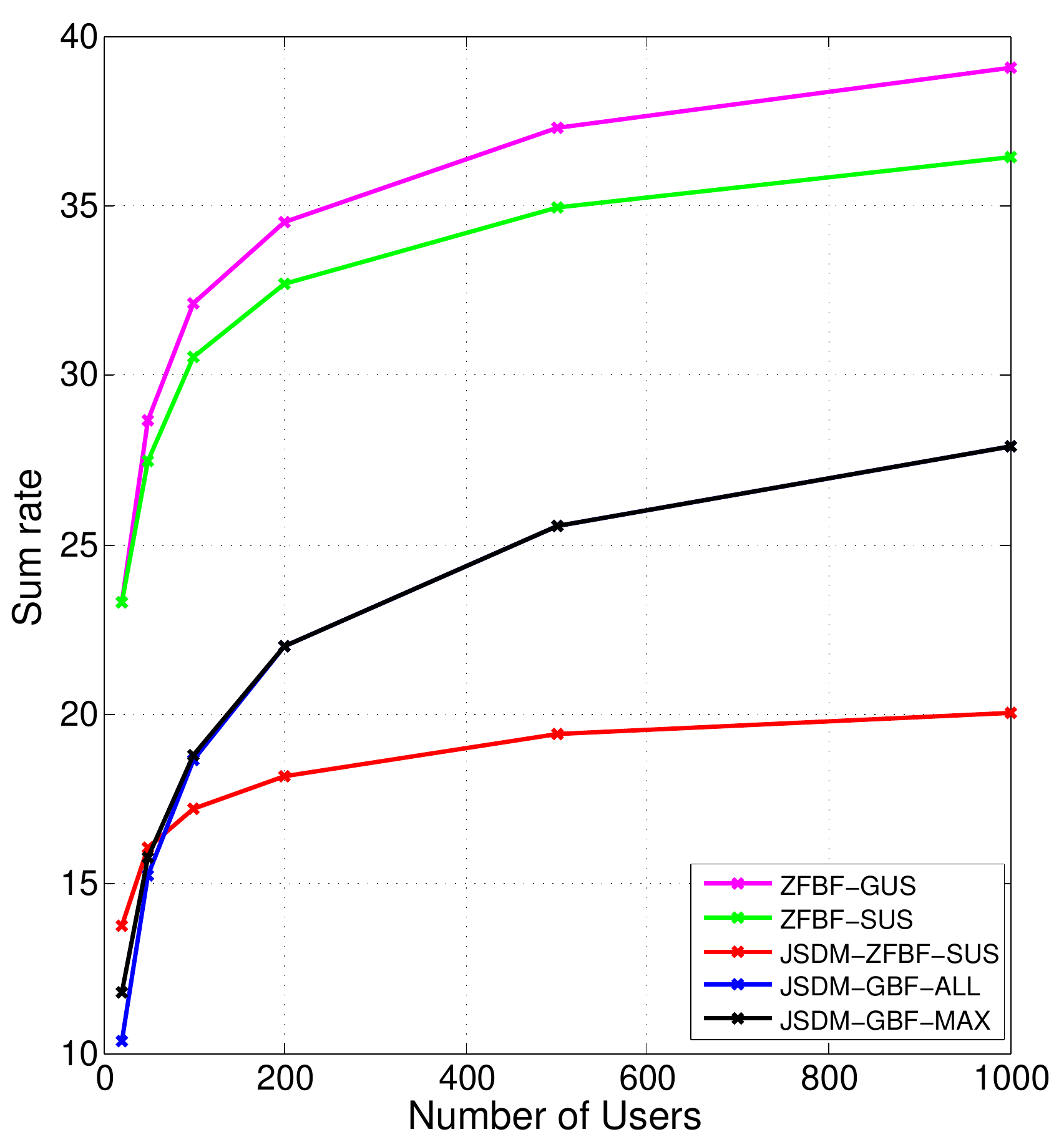}
  \label{fig:M-8}
  }
  \subfigure[$M = 16$]{
  \includegraphics[width=8cm]{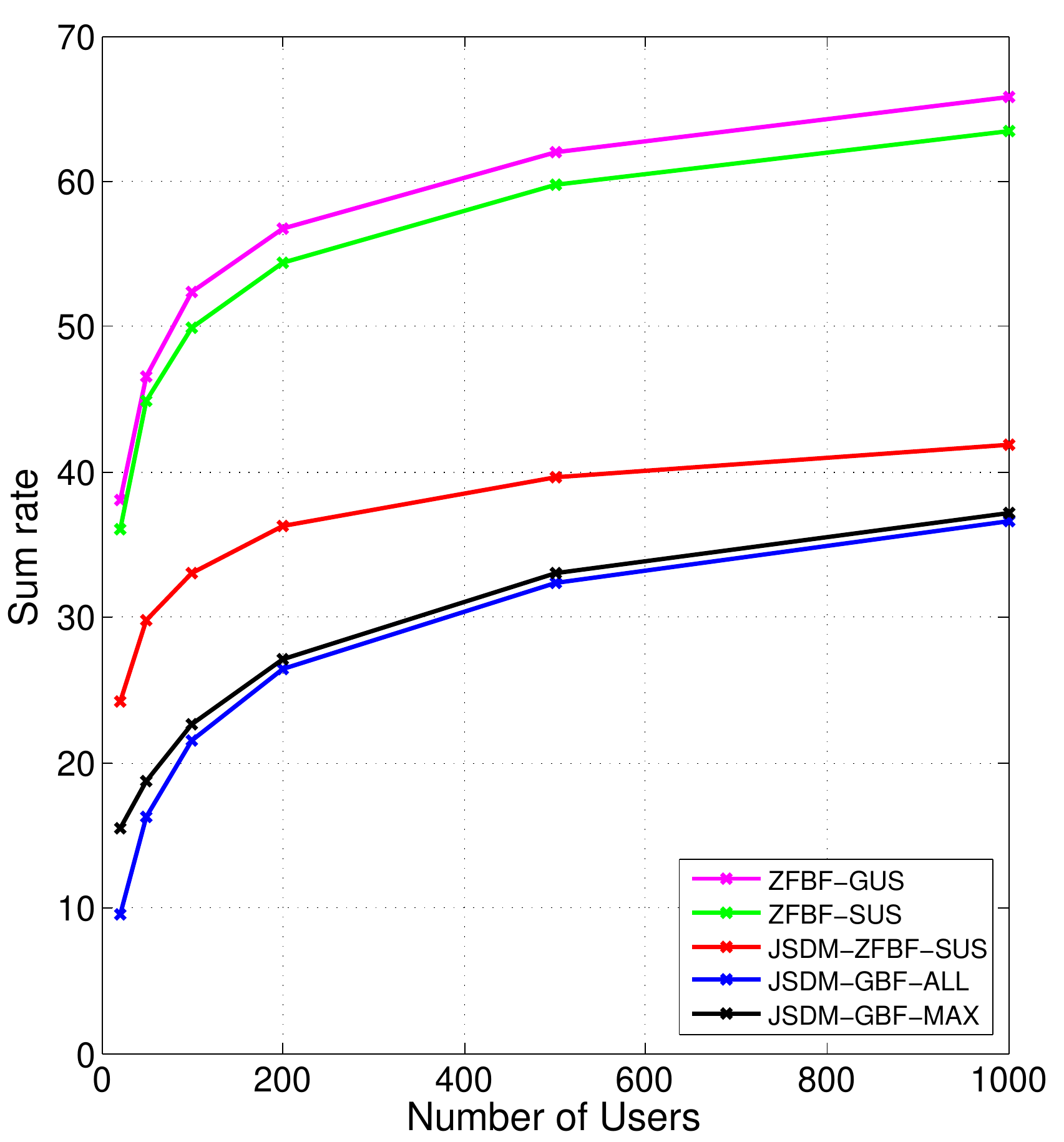}
  \label{fig:M-16}
  }
  \caption{Comparison of sum spectral efficiency (bit/s/Hz) vs. number of users for JSDM with DFT-based fixed quantization
  user grouping and different user selection algorithms.}
  \label{fig:fft-comp}
\end{figure}

Figures \ref{fig:M-8} and \ref{fig:M-16} shows the sum spectral efficiency (in bits/sec/Hz) versus the number of users in the system, averaged over the two patterns with different user selection algorithms for $M = 8$ and $M = 16$ respectively, when DFT-based user grouping is applied.
For the sake of comparison, we show also the performance of
ZFBF with greedy user selection \cite{dimic2005downlink}
(denoted by {\em ZFBF-GUS}) and ZFBF with semi-orthogonal user selection \cite{yoo2006optimality}
(denoted by {\em ZFBF-SUS}), where instead of restricting to JSDM with per-group processing, the
selection is performed across all users without grouping, on the basis of the  full channel state information (i.e., without multiplication by the pre-beamforming matrices).
These performances are shown here to compare how JSDM performs with respect to classical linear beamforming schemes without the structure constraint of
fixed pre-beamforming. Notice that these schemes  require full channel state feedback from all users, and therefore are typically too costly in terms of feedback
in order to be practical.

\begin{figure}
\centering \subfigure[$ZFBF-SUS$]{
  \includegraphics[width=8cm]{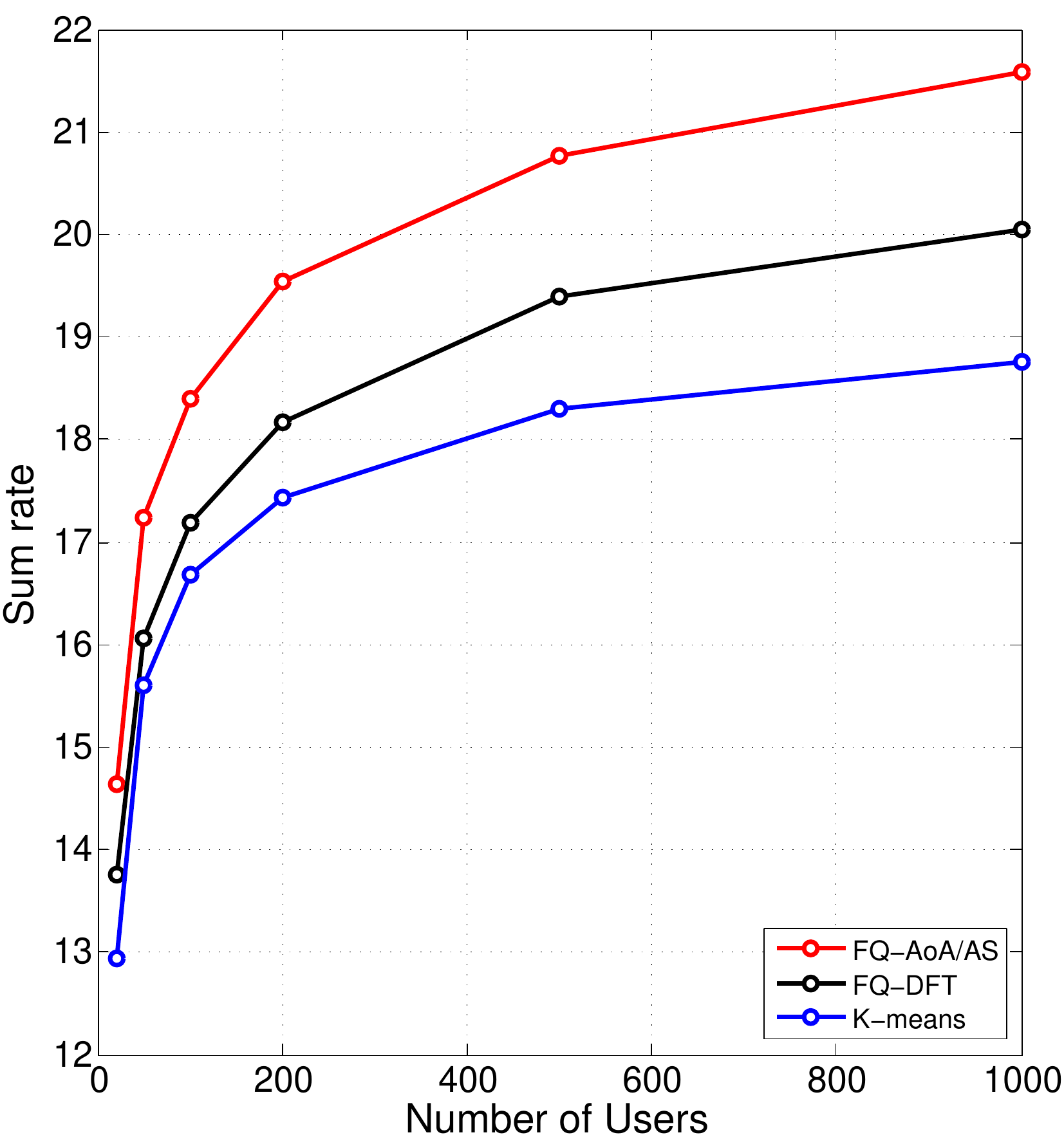}
  \label{fig:M-8-sus}
  }
  \subfigure[$GBF-MAX$]{
  \includegraphics[width=8cm]{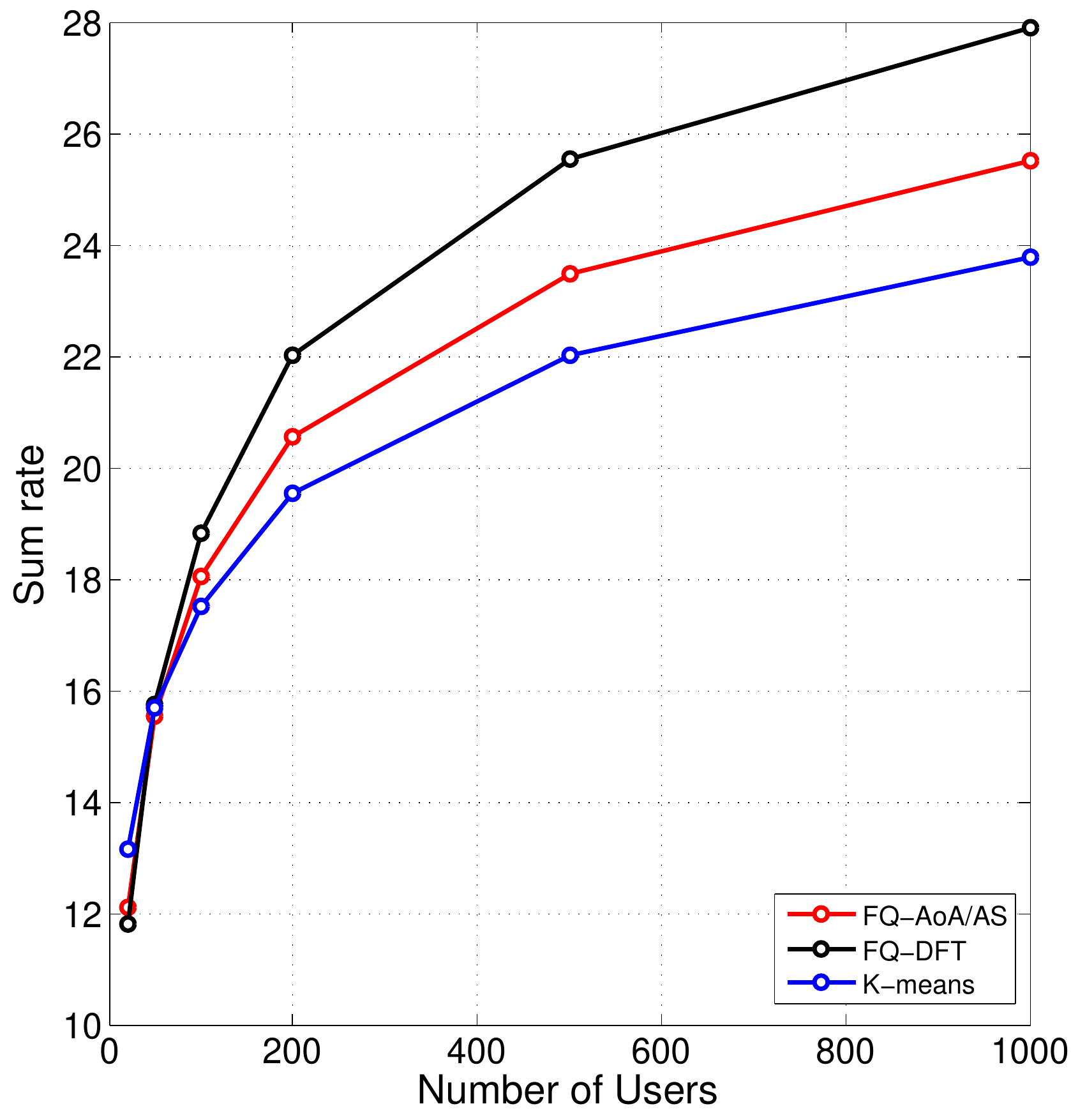}
  \label{fig:M-8-gbf}
  }
  \caption{Comparison of sum spectral efficiency (bit/s/Hz) vs. number of users for JSDM-ZFBF-SUS and JSDM-GBF-MAX with different user grouping  algorithms for $M = 8$.}
  \label{fig:usr-grp-M-8}
\end{figure}

\begin{figure}
\centering \subfigure[$ZFBF-SUS$]{
  \includegraphics[width=8cm]{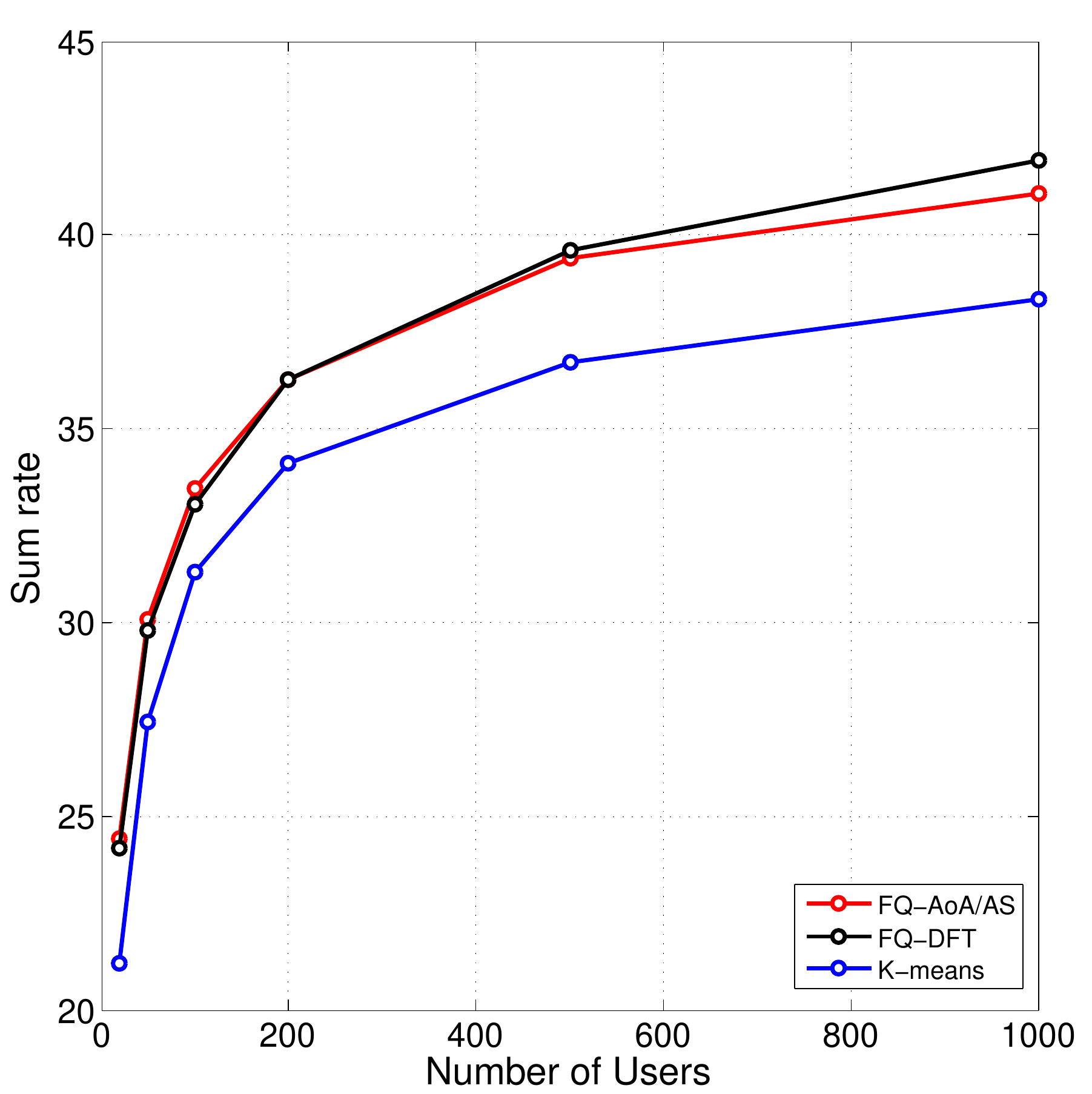}
  \label{fig:M-16-sus}
  }
  \subfigure[$GBF-MAX$]{
  \includegraphics[width=8cm]{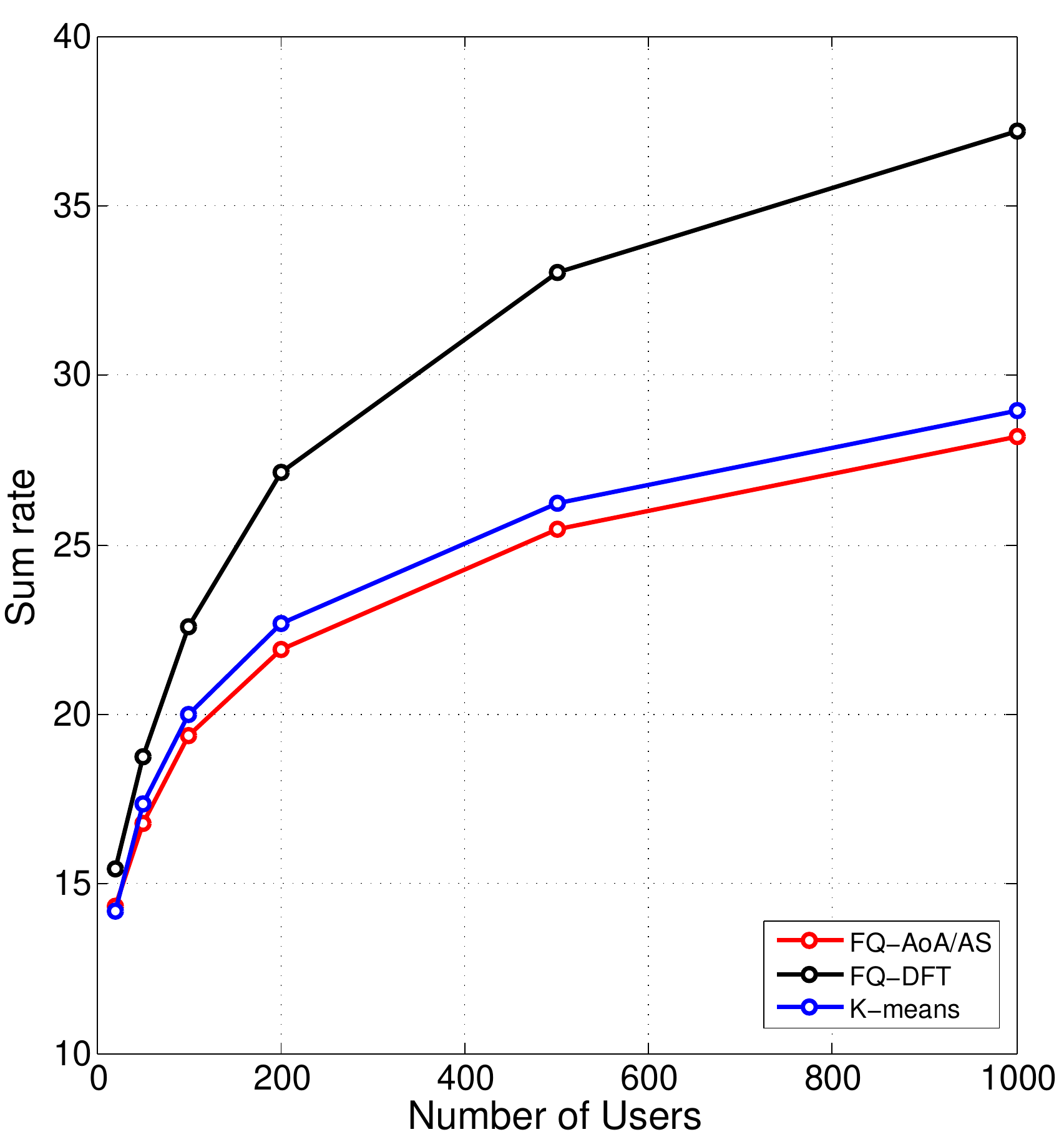}
  \label{fig:M-16-gbf}
  }
  \caption{Comparison of sum spectral efficiency (bit/s/Hz) vs. number of users for JSDM-ZFBF-SUS and JSDM-GBF-MAX with different user grouping  algorithms for $M = 16$.}
  \label{fig:usr-grp-M-16}
\end{figure}

Figure \ref{fig:M-8-sus} shows the sum spectral efficiency versus the number of users for different user grouping algorithms with
$M = 8$ and {\em JSDM-ZFBF-SUS}.
Figure \ref{fig:M-8-gbf} shows analogous results for {\em JSDM-GBF-MAX}.
Figures \ref{fig:M-16-sus} and \ref{fig:M-16-gbf} show the corresponding results for $M = 16$.
These results indicate that user grouping by  DFT-based fixed quantization performs generally better than the other user
grouping algorithm  considered in this work. Hence, because of its simplicity, this appears to be the preferred method for practical
user grouping.

\section{Large System Limit} \label{sec:large-system-limit}

In this section, we focus on the large system limit, i.e., when the number of antennas and the number of users go to infinity with a fixed ratio.
Specifically,  we modify the system model of Section \ref{sec:JSDM} and introduce a parameter $N$,
such that the BS has $MN$ antennas and serves $KN$ single antenna users.
Then, we consider the system performance for $N \rightarrow \infty$.
In this limit, we shall see in Section \ref{sec:user-grouping-simplified} that the user grouping scheme with DEF-based fixed quantization takes on a very simple form,
corresponding to a quantization of the AoA/AS plane. This requires only the knowledge of the AoAs and ASs of the users,
instead of the whole covariance matrix.  As far as user selection is concerned,
we notice from \cite{caire2009selection} that in the limit of $N \rightarrow \infty$ user selection schemes max SINR or SUS) become less and less effective because of
``channel hardening'', and require CQI feedback or effective channel state feedback from a large number of users, such that the benefit of
``opportunistic'' user selection is quickly offset by the extra cost of the feedback.
Hence, following the idea of \cite{huh2012network},  we consider a probabilistic user selection scheme where users are selected with
a certain probability distribution, which is optimized in order to maximize the system ergodic sum rate. in this way, only the users actually scheduled for transmission
will have to feed back their channel state information, in line with the observations made in \cite{ravindran2008multi}.

\subsection{DFT-based user grouping in the large system limit} \label{sec:user-grouping-simplified}

When the number of antennas $MN$ is large, and the BS is equipped with a uniform linear array, the eigenvectors of the channel covariance can be approximated by a subset of the columns of the DFT matrix \cite{adhikary2012joint}.
Thus, $\Um_k$, i.e., the matrix of eigenvectors of the channel covariance of a user $k$ with angle of arrival $\theta_k$ and angular spread $\Delta_k$,
is well approximated  by a matrix $\Fm_k^{\rm u}$, formed by a subset of the columns of the DFT  $\Fm$,
where the subscript denotes the user index $k$ and the superscript is to indicate  that this matrix is specific for user $k$.
In a similar fashion, $\Fm_g^{\rm grp}$ denotes the DFT-based group eigenspace $\Vm_g^*$ for a particular group $g$.

Letting $\Fm = [\fv_{-\frac{MN}{2} + 1}\ \fv_{-\frac{MN}{2} + 2}\ \ldots \fv_{-1}\ \fv_0\ \fv_1\ \ldots \fv_{\frac{MN}{2}}]$ denote the $MN \times MN$ DFT matrix,
the vector $\fv_i$ denotes the Fourier vector corresponding to frequency $\frac{i}{MN}$.
The matrices $\Fm_k^{\rm u}$ and $\Fm_g^{\rm grp}$) are formed by blocks DFT columns corresponding to adjacent frequencies, i.e., they take on the form
\begin{equation}
[\fv_{l} \fv_{l + 1} \ldots \fv_{u}],
\end{equation}
from some interval of DFT frequencies corresponding to the integer indices in $[l, u]$.
From the analysis in \cite{adhikary2012joint}, we know that $\Fm_k^{\rm u}$
for a certain $(\theta_k,\Delta_k)$  contains the DFT frequencies with indices in $\Fc^u_k = [l_k, u_k]$, with
\begin{equation} \label{lu-limits}
\left \{ \begin{array}{l}
l_k = \lfloor -M N D \sin(\theta_k + \Delta_k) \rfloor \\
u_k = \lceil -M N D \sin(\theta_k - \Delta_k) \rceil \end{array}  \right .
\end{equation}
For the group eigenspace, we denote the frequency index interval forming $\Fm_g^{\rm grp}$
by  $\Fc_g^{\rm grp} = \{L_g,L_g+1,\ldots,U_g\}$, for some $L_g, U_g$ suitably defined (see later).

With the notations in place, the fixed quantization grouping scheme assigns user $k$ to group $g$ if
\begin{eqnarray} \label{eqn:simplified-temp}
g &=& {\rm arg} \min_{g'} ||\Um_k^* \Um_k^{* \herm} - \Vm_{g'}^* \Vm_{g'}^{* \herm} ||_F^2 \nonumber\\
&=& {\rm arg} \min_{g'} ||\Fm_{k}^{\rm u} \Fm_{k}^{{\rm u} \herm} - \Fm_{g'}^{\rm grp} \Fm_{g'}^{{\rm grp} \herm}||_F^2 \nonumber\\
&=& {\rm arg} \min_{g'} \left(\bar{r}_k^* + r_{g'}^* - 2 ||\Fm_{k}^{{\rm u} \herm} \Fm_{g'}^{\rm grp} ||_F^2 \right)
\end{eqnarray}
where $\bar{r}_k^*$ and $r_g^*$ denote the number of columns (rank) in $\Fm_{k}^{\rm u}$ and $\Fm_{g'}^{\rm grp}$ respectively.
For the sake of analysis, we assume that $r_{g}^* = r \ \forall\ g$. This reduces the decision rule in (\ref{eqn:simplified-temp}) to
\begin{equation}
g = {\rm arg} \max_{g'} ||\Fm_{k}^{{\rm u} \herm} \Fm_{g'}^{\rm grp} ||_F^2.
\end{equation}
Letting $N \rightarrow \infty$ and normalizing by $1/N$, we obtain
\begin{eqnarray}
\lim_{N \rightarrow \infty} \frac{1}{N} ||\Fm_{k}^{{\rm u} \herm} \Fm_{g'}^{\rm grp} ||_F^2 &=& \lim_{N \rightarrow \infty} \frac{1}{N}  \sum_{m \in \Fc_k^{\rm u}} \sum_{n \in \Fc_g^{\rm grp}} \left| \frac{1}{MN} \sum_{k=-\frac{MN}{2} + 1}^{\frac{MN}{2}} e^{j \frac{2 \pi}{MN} (n - m) k}\right|^2 \nonumber\\
&=& \lim_{N \rightarrow \infty} \frac{1}{N}  \sum_{m \in \Fc_k^{\rm u}} \sum_{n \in \Fc_g^{\rm grp}} \left| \frac{1}{MN} \frac{1 - e^{j \frac{2 \pi}{MN}(n - m) MN}}{1 - e^{j \frac{2 \pi}{MN} (n - m)}} \right|^2 \nonumber\\
&=& \lim_{N \rightarrow \infty} \frac{1}{N}  \sum_{m \in \Fc_k^{\rm u}} \sum_{n \in \Fc_g^{\rm grp}} \delta_{m ,n}
\end{eqnarray}
where $\delta_{m,n}$ is the Kronecker delta function.\footnote{$\delta_{m,n} = \left\{ \begin{array}{cc} 1 & m = n\\ 0 & m\neq n \end{array}\right.$}

For finite $N$, $\sum_{m \in \Fc_k^{\rm u}} \sum_{n \in \Fc_g^{\rm grp}} \delta_{m,n}$ gives the number of identical columns in
$\Fm_{g}^{\rm grp}$ and $\Fm_{k}^{\rm u}$. In the limit $N \rightarrow \infty$, the term $\frac{1}{N} \sum_{m \in \Fc_k^{\rm u}} \sum_{n \in \Fc_g^{\rm grp}} \delta_{m,n}$ reduces to $M \Phi_g^k$, where $\Phi_g^k$ denotes the overlap between the intervals $(\frac{l_k}{MN},\frac{u_k}{MN})$ and $(\frac{L_g}{MN},\frac{U_g}{MN})$.
Hence, (\ref{eqn:simplified-temp}) reduces to
\begin{equation} \label{eqn:overlap-condn}
g = {\rm arg} \max_{g'} \Phi_{g'}^k
\end{equation}
We can further simplify (\ref{eqn:overlap-condn}), leading to a much simpler expression for the user grouping algorithm.
For notational convenience, we make the following change of notations
$a_k = \frac{l_k + u_k}{2MN}, b_k = \frac{u_k - l_k}{2MN}, A_g = \frac{L_g + U_g}{2MN}, B_g = \frac{U_g - L_g}{2MN}$.
The overlap $\Phi_g^k$ between the intervals $(a_k - b_k,a_k + b_k)$ and $(A_g - B_g,A_g + B_g)$ is given by
\begin{equation}
\Phi_g^k = \left\{ \begin{array}{cc} \min(a_k + b_k,A_g + B_g) - \max(a_k - b_k, A_g - B_g) & \min(a_k + b_k,A_g + B_g) - \max(a_k - b_k, A_g - B_g) > 0\\
0 & \min(a_k + b_k,A_g + B_g) - \max(a_k - b_k, A_g - B_g) \leq 0 \end{array}\right.
\end{equation}
Denote $\Delta_{\rm int} = |b_k - B_g|$. We consider two cases:
\begin{itemize}
\item \textbf{Case 1}: $b_k - B_g = \Delta_{\rm int} > 0$.
Considering only the case of $\min(a_k + b_k,A_g + B_g) - \max(a_k - b_k, A_g - B_g) > 0$,  we have
\begin{eqnarray} \label{eqn:case-1}
\Phi_g^k &=& \min(a_k + b_k,A_g + B_g) - \max(a_k - b_k, A_g - B_g) \nonumber\\
&=& \min(a_k + B_g + \Delta_{\rm int},A_g + B_g) - \max(a_k - B_g - \Delta_{\rm int}, A_g - B_g) \nonumber\\
&=& 2 B_g + \min(a_k + \Delta_{\rm int},A_g) - \max(a_k - \Delta_{\rm int}, A_g) \nonumber\\
&=& \left\{ \begin{array}{cc}
2 B_g + A_g - a_k + \Delta_{\rm int} & A_g < a_k - \Delta_{\rm int}\\
2 B_g & a_k - \Delta_{\rm int} \leq A_g \leq a_k + \Delta_{\rm int}\\
2 B_g + a_k - A_g + \Delta_{\rm int} & A_g > a_k + \Delta_{\rm int}
\end{array}\right.
\end{eqnarray}
\item \textbf{Case 2}: $B_g - b_k = \Delta_{\rm int} > 0$
Proceeding similarly as Case 1, we get
\begin{eqnarray} \label{eqn:case-2}
\Phi_g^k &=& 2 b_x + \min(a_k,A_g + \Delta_{\rm int}) - \max(a_k, A_g - \Delta_{\rm int}) \nonumber\\
&=& \left\{ \begin{array}{cc}
2 b_k + A_g - a_k + \Delta_{\rm int} & a_k > A_g + \Delta_{\rm int}\\
2 b_k & A_g - \Delta_{\rm int} \leq a_k \leq A_g + \Delta_{\rm int}\\
2 b_k + a_k - A_g + \Delta_{\rm int} & a_k < A_g - \Delta_{\rm int}
\end{array}\right.
\end{eqnarray}
\end{itemize}

From (\ref{eqn:case-1}) and (\ref{eqn:case-2}), it is easy to see that, for a given $k$,
\begin{equation} \label{eqn:overlap-2}
\max_g \Phi_g^k = \min_g |A_g - a_k|
\end{equation}
Using expressions (\ref{lu-limits}) for $l_k$ and $u_k$ as $N \rightarrow \infty$ we have:
\begin{eqnarray}
a_k &=& \lim_{N \rightarrow \infty} = \frac{l_k + u_k}{MN} \nonumber\\
&=& \frac{-D \sin(\theta_k + \Delta_k) + (-D \sin(\theta_k - \Delta_k))}{2} \nonumber\\
&=& -D \sin(\theta_k) \cos(\Delta_k),
\end{eqnarray}
reducing (\ref{eqn:overlap-2}) to
\begin{equation} \label{eqn:overlap-final}
\max_g \Phi_g^k = \min_g |A_g - (-D \sin(\theta_k) \cos(\Delta_k))|.
\end{equation}

\begin{example} \label{example:user-grouping}
Let us assume $\theta_k \in (-60^o,60^o)$, $\Delta_k \in (5^o,15^o)$ and $G = 4$. $L_g = \lfloor M N \frac{g-1}{G} - \frac{MN}{2}\rfloor$ and $U_g = \lceil M N \frac{g}{G} - \frac{MN}{2} \rceil$. Thus, we have $A_g = \frac{g - \frac{1}{2}}{G} - \frac{1}{2}$. Figure \ref{fig:cluster-ex-limit} shows the partition of the $(\theta_k,\Delta_k)$ space into $G = 4$ groups using (\ref{eqn:overlap-final}). The different colors indicate the different groups.

\begin{figure}[ht]
\centerline{\includegraphics[width=8cm]{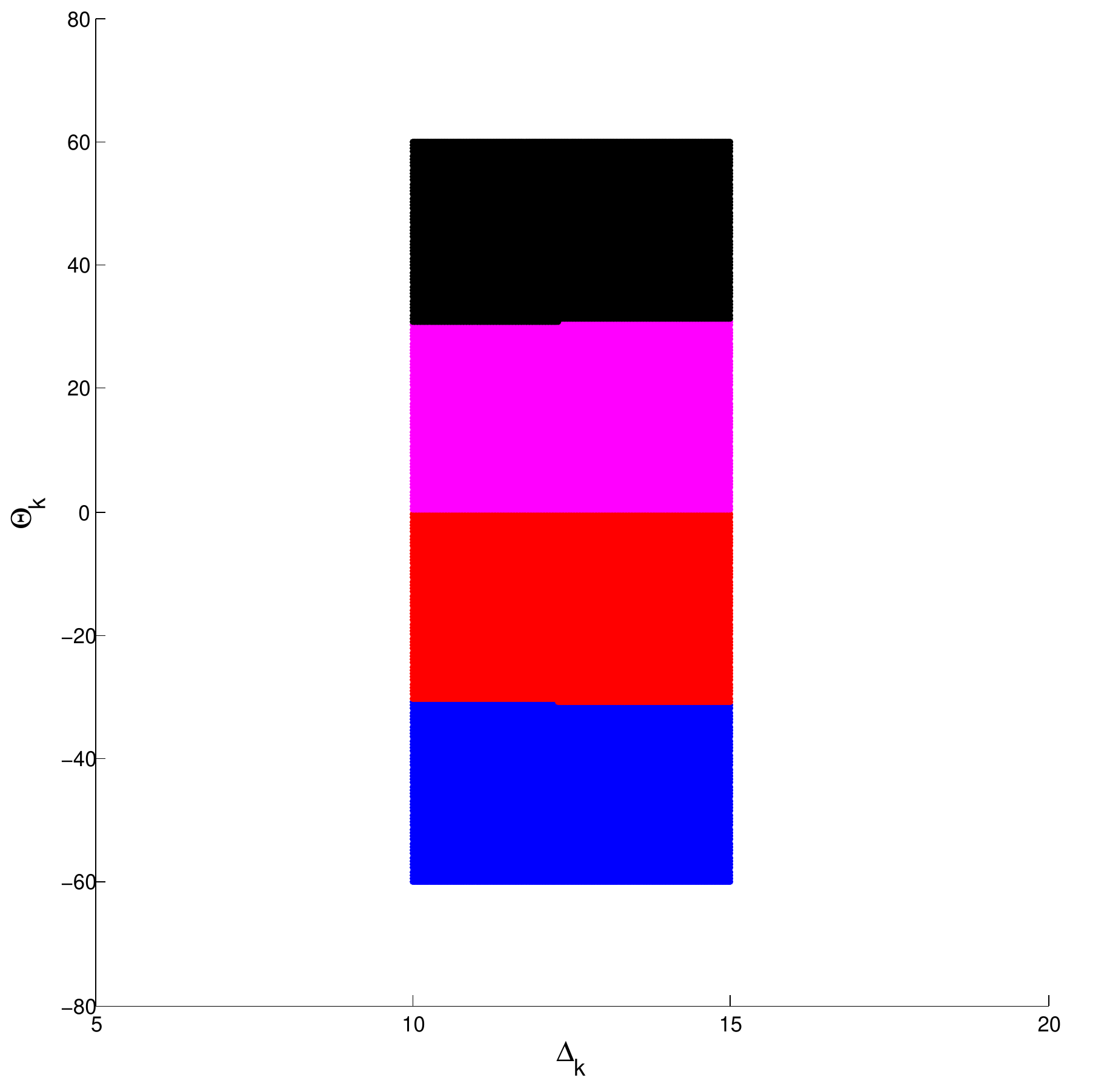}} \caption{Illustration of user grouping in the large system limit. 'black' denotes $g = 1$, 'magenta' denotes $g = 2$ and 'red' and 'blue' denote $g = 3$ and $4$.} \label{fig:cluster-ex-limit}
\end{figure}
\end{example}

\subsection{Probabilistic user selection}  \label{sec:prob-user-selection}

We assume that users are arranged into  $K$ co-located ``user subgroups'' with $N$ users in each location.
Notice that this assumption is made for analytical convenience, and corresponds to the quantization of the user spatial distribution into
a number of discrete points in the coverage area. In this case, if the user locations correspond, on average, to an area $A = \frac{\mbox{total coverage area}}{K} ($m$^2$),
then $N/A$ is the user average density (users/m$^2$). For a fixed coverage area, both the number of BS antennas and the user density
grows to infinity, such that the total number of users per BS antenna is fixed and equal to $K/M$.  Users in each subgroup $k$
are statistically equivalent,  with common covariance matrix $\Rm_k(\theta_k,\Delta_k)$ that depends only on the location (AoA) and local scattering (AS).
We define a group as a collection of subgroups, obtained by the application of the user grouping algorithm.
The number of subgroups forming a group $g$ is indicated by $K_g$, such that $K = \sum_g K_g$.
Defining the $M N \times N$ channel matrix of a user subgroup $g_k$, i.e., the $k$-th subgroup of the $g$-th group,
as $\bar{\Hm}_{g_k}$, we have $\Hm_g$, the effective channel matrix of group $g$ as $\Hm_g = [\bar{\Hm}_{g_1} \ldots \bar{\Hm}_{g_{K_g}}]$.
As a result, the received signal vector for users in subgroup $g$ served by JSDM with PGP can be written as
\begin{equation}
\yv_g = \begin{pmatrix} \bar{\Hm}_{g_1}^\herm \Bm_g\\
 \vdots\\
\bar{\Hm}_{g_{K_g}}^\herm \Bm_g \end{pmatrix} \Pm_g \dv_g + \left(\sum_{g'=1,g' \neq g}^G \begin{pmatrix} \bar{\Hm}_{g_1}^\herm \Bm_{g'}\\
 \vdots\\
\bar{\Hm}_{g_{K_g}}^\herm \Bm_{g'} \end{pmatrix} \Pm_{g'} \dv_{g'} \right) + \zv_g
\end{equation}
Note that $\Bm_g$ is $MN \times b_gN$, $\Hm_g$ is $MN \times K_g N$ and $\Pm_g$ is $b_gN \times S_gN$, i.e., the dimensions of all matrices
scale linearly with $N$.

The BS now needs to perform downlink scheduling, i.e., selecting a subset of $S_g N \leq \min(b_g N, K_g N)$  out of $K_g N$ users in order to serve them in a particular time-frequency slot. The scheduling problem is formulated as the maximization of a strictly increasing and concave
{\em Network Utility Function} $\Gc(\cdot)$, representing some notion of ``fairness'' over the set of achievable ergodic rates, or ``throughputs''.
Since users within a subgroup are statistically equivalent, they all have the same throughput.
Denoting the normalized~\footnote{Since the number of downlink streams scales linearly with $N$, the sum rate per group also
increases linearly in $N$. Then, in order to obtain meaningful limits and a meaningful
network utility maximization problem, we divide the sum rate per group by $N$.}
throughput of a subgroup $k$ as $\bar{R}_k$, i.e., the sum rate of the users in subgroup $k$ divided by $N$,
the scheduling problem reduces to solving the following optimization:
\begin{eqnarray} \label{eqn:optz}
{\rm maximize} && \Gc(\bar{R}_1,\ldots,\bar{R}_K)\nonumber\\
{\rm subject\ to} && \bar{R}_k \in \bar{\Rc}\  \forall\ k \in \{1,2,\ldots,K\}
\end{eqnarray}
where $\bar{\Rc}$ is the system throughput region.
For the sake of simplicity, we focus on ZFBF multiuser precoding  with equal power allocated to each downlink data  stream, such that
$\bar{\Rc}$ denotes the throughput region obtained  using such scheme.

With zero forcing precoding, the BS can serve a maximum of $\min\{K_g N, b_g N\}$ users in a given time frequency slot in group $g$.
In realistic scenarios of interest, $K_g > b_g$ for all groups $g$.
This condition can be always verified assuming that the system has a sufficiently large number of users and
by discretizing sufficiently finely the user spatial distribution such that the number of co-located subgroups in each group is large enough.\footnote{If
the number of users is not large enough,  the system achievable sum rate is limited by the number of users,  i.e., the system is not fully loaded.
Instead, here we are interested in a fully loaded system where the biting constraint is given by the number of possible
downlink streams, not by the number of users. Hence, we assume to have a large enough number of users in order to exploit all the available JSDM
multiplexing gain.}
Therefore, the BS should select a subset of users, not larger than $\sum_{g=1}^G b_g N$ to be served
at each time slot.
In order to find the optimal subset, an exhaustive search over all possible subsets of size less than or equal to $\sum_{g=1}^G b_g N$ over $KN$ users
is required. This search is combinatorial and it would be infeasible for a large number of users.
Furthermore, it requires all users to feed back their effective channel vectors, which may be very costly when
the number of users is much larger than the number of possible downlink streams, i.e., the number of {\em active} users on each time-frequency slot.
Hence, following the idea proposed in \cite{huh2012network}, we make the following simplifying assumption in order to make the problem
in (\ref{eqn:optz}) analytically tractable and, at the same time, obtain a practical probabilistic user selection scheme
that achieves good results at moderate complexity and requires only the active users (users effectively scheduled) to feed back their channel state information.
In this scheme,  the BS selects $\gamma_kN$ users in each subgroup $k$ by random selection, independent of the channel matrix realization, where
$\gamma_k \in [0,1]$ is referred to as the ``fraction'' of active users.
Notice that while $\gamma_k$ is a constant that depends only on the system statistics,
the set of active users changes randomly at each slot and is uniformly distributed over all
possible ${ N \choose {\gamma_kN}}$ active user subsets. Then, only these selected users feed back their effective channel and
the BS serves them using JSDM-ZFBF (i.e., ZFBF in each group, irrespectively of the inter-group interference).

Denoting by $P^u$ the per-stream downlink power (same for all streams)
and noticing that the number of downlink streams per group is $S_g N = \sum_{k =1}^{K_g} \gamma_k N$, we have the constraints
\begin{eqnarray}
S_g = \sum_{k =1}^{K_g} \gamma_{g_k} & \leq & b_g \;\; \forall \; g = 1,\ldots, G   \label{bg-constraint} \\
N P^u  \sum_{g=1}^G S_g & \leq & P,    \label{total-power-contstraint}
\end{eqnarray}
where $P$ is the total power of the BS.
The channel matrices $\bar{\Hm}_{g_k}$ are now functions of $\gamma_{g_k}$, i.e., $\bar{\Hm}_{g_k} = \bar{\Hm}_{g_k}(\gamma_{g_k})$
and have dimension $M N \times \gamma_{g_k} N$.
In order to find an expression for the instantaneous rate of a generic $n$-th active user in
subgroup $k$ of group $g$, we need to to find a convenient asymptotic expression for its SINR, given by
\begin{equation} \label{eqn:SINR}
{\rm SINR}_{g_k}^n = \frac{P^u | \bar{\hv}_{g_k,n}^\herm \Bm_g \Km_g \Bm_g^\herm \bar{\hv}_{g_k,n} |^2 }
{1 + \sum_{g' =1,g' \neq g}^G \sum_{l \in \Sc_{g'}}\sum_{m = 1}^N P^u | \bar{\hv}_{g_k,n}^\herm \Bm_{g'} \Km_{g'} \Bm_{g'}^\herm \bar{\hv}_{g'_l,m} |^2}
\end{equation}
where $\bar{\hv}_{g_k,m}$ denotes the $m$-th column of $\bar{\Hm}_{g_k}$, the ZFBF precoding matrix for group $g$ is given by
\[ \Pm_g = \zeta_g \Bm_g^\herm \Hm_g \left ( \Hm_g^\herm \Bm_g \Bm_g^\herm \Hm_g \right )^{-1} \]
where $\zeta_g$ is a power normalization factor to be defined later, and where we let
 \begin{equation}
 \Km_g = \zeta^2_g \Bm_g^\herm \Hm_g \left(\Hm_g^\herm \Bm_g \Bm_g^\herm \Hm_g \right)^{-2} \Hm_g^\herm \Bm_g,
\end{equation}
Notice that in the SINR denominator we have only the contribution of the inter-group interference, since the intra-group interference is completely eliminated by
ZFBF precoding.
In order to determine $\zeta_g$, notice that the total transmit power for group $g$ is given by $P^u NS_g$, with
$S_g = \sum_{k=1}^{K_g} \gamma_{g_k}$ as defined before,  and where $NS_g$ is the number of downlink data streams served to group $g$ using spatial multiplexing.
Then, we have
\begin{eqnarray}
\trace{\left(\Bm_g \Pm_g \EE[ \dv_g \dv_g^\herm] \Pm_g^\herm \Bm_g^\herm \right)} &=& P^u N S_g \nonumber\\
\Longrightarrow \zeta^2 \trace{\left(\Bm_g \Bm_g^\herm \Hm_g \left(\Hm_g^\herm \Bm_g \Bm_g^\herm \Hm_g \right)^{-2} \Hm_g^\herm \Bm_g \Bm_g^\herm \right)}
&=& N S_g \nonumber\\
\Longrightarrow \zeta_g^2 &=& \frac{N S_g}{\trace{\left(\Bm_g \Bm_g^\herm \Hm_g \left(\Hm_g^\herm \Bm_g \Bm_g^\herm \Hm_g \right)^{-2} \Hm_g^\herm \Bm_g \Bm_g^\herm \right)}}. \label{zetaminchia}
\end{eqnarray}
When $\Bm_g$ is tall unitary (e.g., with the choice $\Bm_g = \Fm^{\rm grp}_g$ discussed before, (\ref{zetaminchia}) simplifies to
\[ \zeta_g^2 = \frac{N S_g}{\trace{\left(\Hm_g^\herm \Bm_g \Bm_g^\herm \Hm_g \right)^{-1} }}. \]
In the limit of $N \rightarrow \infty$, the terms ${\rm SINR}_{g_k}^n$ for all users in the same subgroup converge to the same deterministic quantity,
that depends only on the subgroup index $g_k$ \cite{wagner2012large}, \cite{huh2012network}
\begin{equation}
{\rm SINR}_{g_k}^n \stackrel{N \rightarrow \infty}{\longrightarrow} {{\rm SINR}_{g_k}^o}
\end{equation}
As a result, the achievable normalized throughput  fora  subgroup $k$ of group $g$  is given by $\bar{R}_{g_k} = \gamma_{g_k} \log (1 + {{\rm SINR}_{g_k}^o})$,
reducing the optimization problem (\ref{eqn:optz}) to
\begin{eqnarray} \label{eqn:optz-2}
{\rm maximize} && \Gc\left ( \left \{\gamma_{g_k} \log (1 + {{\rm SINR}_{g_k}^o}) \; : \; g  = 1,\ldots,G, \; k = 1,\ldots, K_g \right \} \right )\nonumber \\
{\rm subject\  to} && S_g = \sum_{k =1}^{K_g} \gamma_{g_k} \leq b_g \nonumber\\
&& 0 \leq \gamma_{g_k} \leq 1 \ \forall \ g = 1,\ldots,G, \ k  = 1, \ldots, K_g.
\end{eqnarray}
with respect to the optimization variables $\{\gamma_{g_k}\}$, which define how many downlink streams should be dedicated to
each user subgroup by the JSDM scheme.

In the next section, we provide methods to calculate ${{\rm SINR}_{g_k}^o}$ using the technique
of \cite{wagner2012large} and propose a greedy user selection algorithm to calculate the fractions $\gamma_{g_k}$
in order to provide a good heuristic solution for the non-convex problem (\ref{eqn:optz-2}), which mimics the well-known
greedy user selection for the combinatorial optimization of the finite-dimensional sum rate of the ZFBF multiuser precoding \cite{dimic2005downlink}.

\subsection{Analysis}  \label{asym-analysis}

Following the analysis technique developed in \cite{wagner2012large}, and widely used in our previous work \cite{adhikary2012joint},
the sought SINR limit ${{\rm SINR}_{g_k}^o}$ is given by
\begin{equation}  \label{limit-SINR}
{{\rm SINR}_{g_k}^o} = \frac{{\zeta_g^o}^2 P/S}{1 + \sum_{g' = 1,g' \neq g}^G {\zeta_{g'}^o}^2 \Upsilon_{g',g_k}^o P/S},
\end{equation}
where $P$ indicates the total transmit power and we define $S = \sum_{g=1}^G S_g$ to be the total number of downlink data streams across all groups,
normalized by $N$.  The expressions of $\Upsilon_{g',g_k}^o$ and ${\zeta_g^o}^2$ are obtained by a sequence of converging approximations for increasing (finite) $N$,
given as follows:
\begin{eqnarray}
\Upsilon_{g',g_k} &=&  \frac{1}{b_g} \sum_{l = 1}^{K_{g'}} \frac{\gamma_{g'_l} n_{g_k,g'_l}^o}{(m_{g'_l}^o)^2} \nonumber\\
\nv'_{g_k} &=& \left( \Id_{K_{g'}} - \Jm_{g'} \right)^{-1} \vv'_{g_k} \nonumber\\
\{\Jm_g\}_{i,j} &=& \frac{\frac{\gamma_{g_i}}{b_g} \trace(\bar{\Rm}_{g_i} \Tm_g \bar{\Rm}_{g_j} \Tm_g)}{N b_g (m_{g_k}^o)^2} \nonumber\\
\{\vv'_{g_k}\}_i &=& \frac{1}{N b_g} \trace(\bar{\Rm}_{g'_i} \Tm_{g'} \Bm_{g'}^\herm \Rm_{g_k} \Bm_{g'} \Tm_{g'}) \nonumber\\
{\zeta_g^o}^2 &=& \frac{S_g}{\Gamma_g^o} \nonumber\\
\Gamma_g^o &=&  \frac{1}{b_g} \sum_{k  =1}^{K_g} \frac{\gamma_{g_k} q_{g,k}}{(m_{g_k})^2}  \nonumber\\
\qv_g &=& \left( \Id_{K_g} - \Jm_g \right)^{-1} \vv_{g} \nonumber\\
\{\vv_{g}\}_i &=& \frac{1}{N b_g} \trace(\bar{\Rm}_{g_i} \Tm_g \Bm_g^\herm \Bm_g \Tm_g)
\end{eqnarray}
where $m_{g_k}^o$ is given by the solution of a set of fixed point equations
\begin{eqnarray}
m_{g_k}^o &=& \frac{1}{N b_g} \trace(\bar{\Rm}_{g_k} \Tm_g)\nonumber\\
\Tm_g &=& \left( \Id_{b_g N} + \frac{1}{b_g} \sum_{l = 1}^{K_g} \frac{\gamma_{g_l}}{m_{g_l}^o} \bar{\Rm}_{g_l} \right)^{-1},
\end{eqnarray}
and $\bar{{\Rm}}_{g_k} = \Bm_g^\herm {\Rm}_{g_k} \Bm_g \ \forall \ g = 1,\ldots,G, \; k = 1,\ldots, K_g$,
$\nv'_{g_k} = [n_{g_k,g'_1}^o,\ldots,n_{g_k,g'_{K_{g'}}}^o]^T$, $\qv_g = [q_{g,1},\ldots,q_{g,K_g}]^T$.

From \cite{adhikary2012joint}, we have that as $N \rightarrow \infty$, the covariance matrix $\Rm_{g_k}$,
for the users in the $k$-th subgroup of the $g$-th group can be approximated as
\begin{equation}
\Rm_{g_k} = \Fm_{g_k} \bar{\Lambdam}_{g_k} \Fm_{g_k}^\herm
\end{equation}
where $\Fm_{g_k}$ is composed of a subset of columns of the DFT matrix $\Fm$, and $\bar{\Lambdam}_{g_k}$ is a diagonal matrix containing the eigenvalues obtained by using the Toeplitz-Circulant approximation (see   \cite{adhikary2012joint}  for details).
Denoting the AoA and AS for this specific subgroup as $\theta_{g_k}$ and  $\Delta_{g_k}$, respectively, we have
\begin{equation} \label{eqn:eigvec-lim}
\Fm_{g_k} = \left[ \fv_{l_{g_k}} \fv_{l_{g_k} + 1} \ldots \fv_{u_{g_k}} \right]
\end{equation}
with
\[ \left \{ \begin{array}{l}
l_{g_k} = \lfloor -M N D \sin(\theta_{g_k} + \Delta_{g_k}) \rfloor \\
u_{g_k} = \lceil -M N D \sin(\theta_{g_k} - \Delta_{g_k}) \rceil . \end{array} \right . \]
Also, $\bar{\Lambdam}_{g_k}$ is given as \cite{adhikary2012joint}
\begin{equation} \label{eqn:eigval-lim}
\bar{\Lambdam}_{g_k} = \frac{1}{2 \Delta_{g_k}} {\rm diag} \left( \frac{1}{\sqrt{D^2 - (l_{g_k}/MN)^2}}, \frac{1}{\sqrt{D^2 - ((l_{g_k} + 1)/MN)^2}}, \ldots, \frac{1}{\sqrt{D^2 - (u_{g_k}/MN)^2}} \right)
\end{equation}
Since we choose the group subspaces to be the span of blocks of mutually orthogonal columns of the DFT matrix $\Fm$,
there is no need for BD and we just  consider DFT pre-beamforming $\Bm_{g'} = \Fm_{g'}^{\rm grp}$.
Hence, for large $N$ we have
\begin{eqnarray}
\Bm_{g'}^\herm \Rm_{g_k} \Bm_{g'} &=& \Fm_{g'}^{\rm grp \herm} \Fm_{g_k} \bar{\Lambdam}_{g_k} \Fm_{g_k}^\herm \Fm_{g'}^{\rm grp} \nonumber\\
&=& \begin{split} \frac{1}{2 \Delta_{g_k}} {\rm diag} \left( \underbrace{0,\ldots,0}_{L_{g'} - \max(L_{g'},l_{g_k}) - 1}, \right. &\left. \frac{1}{\sqrt{D^2 - (\max(L_{g'},l_{g_k})/MN)^2}}, \frac{1}{\sqrt{D^2 - ((\max(L_{g'},l_{g_k}) + 1)/MN)^2}}, \ldots, \right. \\ &\left. \frac{1}{\sqrt{D^2 - (\min(U_{g'},u_{g_k})/MN)^2}}, \underbrace{0,\ldots,0}_{U_{g'} - \min(U_{g'},u_{g_k}) - 1}\right) \end{split}
\end{eqnarray}
In the limit of large $N$,  the set $\{\frac{\max(L_{g'},l_{g_k})}{MN}, \frac{\max(L_{g'},l_{g_k})+1}{MN}, \ldots, \frac{\min(U_{g'},u_{g_k})}{MN} \}$
corresponds to an interval $(a_{g_k}^{g'},b_{g_k}^{g'}) \subset (0,1)$ on the DFT ``frequency'' axis, which becomes the continuous interval $(-\frac{1}{2},\frac{1}{2})$ for $N \rightarrow \infty$.
Define a function $f(g_k,g',x)$ for $x \in (-\frac{1}{2},\frac{1}{2})$ corresponding to the terms of the form $\Bm_{g'}^\herm \Rm_{g_k} \Bm_{g'}$. We have
\begin{equation} \label{fff}
f(g_k,g',x) = \left\{ \begin{array}{cc}
\frac{1}{2 \Delta_{g_k}} \frac{1}{\sqrt{D^2 - x^2}} & x \in (a_{g_k}^{g'},b_{g_k}^{g'}) \\
0 & {\rm elsewhere}
\end{array}\right.
\end{equation}
Based on (\ref{fff}), expressions involving the trace of $\Bm_{g'}^\herm \Rm_{g_k} \Bm_{g'}$ or functions of $\Bm_{g'}^\herm \Rm_{g_k} \Bm_{g'}$
reduce to one-dimensional integrals over the interval $(-\frac{1}{2},\frac{1}{2})$.  For the sake of clarity, consider the following example:
\begin{equation}
\lim_{N \rightarrow \infty} \frac{1}{N b_g} \trace(\bar{\Rm}_{g_k}) = \int_{-\frac{1}{2}}^{\frac{1}{2}} f(g_k,g,x) dx.
\end{equation}
Following this observation, we arrive at a much simplified set of equations to calculate ${{\rm SINR}_{g_k}^o}$ in (\ref{limit-SINR}),
given directly in terms of the limit for $N \rightarrow \infty$, and not just as a sequence of convergent approximations for increasing $N$ as before.
In this case,  $\Upsilon_{g',g_k}^o$ and ${\zeta_g^o}^2$ in (\ref{limit-SINR}) are given by
\begin{eqnarray}
\Upsilon_{g',g_k} &=& \frac{1}{b_g}  \sum_{l  =1}^{K_{g'}} \frac{\gamma_{g'_l} n_{g_k,g'_l}^o }{(m_{g'_l}^o)^2} \nonumber\\
\nv'_{g_k} &=& \left( \Id_{K_{g'}} - \Jm_{g'} \right)^{-1} \vv'_{g_k} \nonumber\\
\{\Jm_g\}_{i,j} &=& \frac{\frac{\gamma_{g_i}}{b_g} \int_{-\frac{1}{2}}^{\frac{1}{2}} \frac{f(g_i,g,x) f(g_j,g,x)}{h(g,x)^2} dx}{(m_{g_k}^o)^2} \nonumber\\
\{\vv'_{g_k}\}_i &=& \int_{-\frac{1}{2}}^{\frac{1}{2}} \frac{f(g'_i,g',x) f(g_k,g',x)}{h(g',x)^2} dx \nonumber\\
{\zeta_g^o}^2 &=& \frac{S_g}{\Gamma_g^o} \nonumber\\
\Gamma_g^o &=& \frac{1}{b_g} \sum_{k =1}^{K_g} \frac{\gamma_{g_k} }{(m_{g_k})^2} q_{g,k} \nonumber\\
\qv_g &=& \left( \Id_{K_g} - \Jm_g \right)^{-1} \vv_{g} \nonumber\\
\{\vv_{g}\}_i &=& \int_{-\frac{1}{2}}^{\frac{1}{2}} \frac{f(g_i,g,x)}{h(g,x)^2} dx
\end{eqnarray}
and $m_{g_k}^o$ given by the solution of a set of fixed point integral equations
\begin{eqnarray}
m_{g_k}^o &=& \int_{-\frac{1}{2}}^{\frac{1}{2}} \frac{f(g_k,g,x)}{h(g,x)} dx \nonumber\\
h(g,x) &=& 1 + \frac{1}{b_g} \sum_{l =1}^{K_g} \frac{\gamma_{g_l}}{m_{g_l}^o} f(g_l,g,x),
\end{eqnarray}
Having an efficient method for calculating the users SINR for given total power $P$ and user group ``fractions'' $\{\gamma_{g_k}\}$,
we propose the following greedy approach to find a good heuristic solution to the non-convex
optimization problem (\ref{eqn:optz-2}), where optimization is with respect to the variables $\{\gamma_{g_k}\}$.

\paragraph{Greedy Algorithm for optimizing the user fractions $\gamma_{g_k}$} \label{algo:greedy-usr-sel}
The greedy algorithm considers incrementing the user fractions in small steps $\delta \gamma$, until the objective function cannot be increased further. We start with $\gamma_{g_k} = 0$ for all subgroups within every group and find $g_k$ such that incrementing the user fraction $\gamma_{g_k}$ by $\delta \gamma$ yields the largest possible increase in the objective function. This procedure is repeated until the objective function cannot be increased further. We denote an iteration by $i$ and, with some abuse of notation,  the objective function as $\Gc(\gammav)$, where $\gammav$ is the vector of all optimization variables $\{\gamma_{g_k}\}$.
\begin{itemize}
\item \textbf{Step 1}: Initialize $i = 0$, $\gamma_{g_k}^{(i)} = 0 \ \ \forall \; g =1,\ldots,G, \; k  = 1, \ldots, K_g$ and $\Gc(\gammav^{(i)}) = \Gc(\zerov)$.
\item \textbf{Step 2}: For $\delta \gamma \ll 1$, set $\gammav_{g_k} = \gammav^{(i)} + \delta \gamma \ev_{g_k}$, where $\ev$ is a vector containing all zeros but a 1 corresponding to the $k$-th subgroup of group $g$. Obtain the corresponding value of the objective function $\Gc(\gammav_{g_k})$  for all pairs
$(g, k) : g  = 1,\ldots,G, \; k = 1,\ldots, K_g$ such that $\gamma_{g_k} \leq 1$ and $\sum_{k =1}^{K_g} \gamma_{g_k} \leq b_g\ \ \forall g$.
For the pairs $(g,k)$ for which the conditions are not satisfied, set $\Gc(\gammav_{g_k}) = \Gc(\gammav^{(i)})$.
If no such pair can be found, then set $\gammav = \gammav^{(i)}$ and exit the algorithm.
\item \textbf{Step 3}: Compute $(\hat{g},\hat{k}) = {\rm arg} \max_{g  = 1,\ldots,G, \; k = 1, \ldots, K_g} \Gc(\gammav_{g_k})$ and set $\Gc(\gammav^{(i+1)})
= \Gc(\gammav_{\hat{g}_{\hat{k}}})$ and $\gammav^{(i+1)} = \gammav_{\hat{g}_{\hat{k}}}$.
\item \textbf{Step 4}: If $\Gc(\gammav^{(i+1)}) > \Gc(\gammav^{(i)})$, increment $i$ by 1 and go to Step 2,
else set $\gammav = \gammav^{(i)}$ and exit the algorithm.
\end{itemize}

\subsection{Results}

\begin{figure}
\centering \subfigure[Pattern 1]{
  \includegraphics[width=8cm]{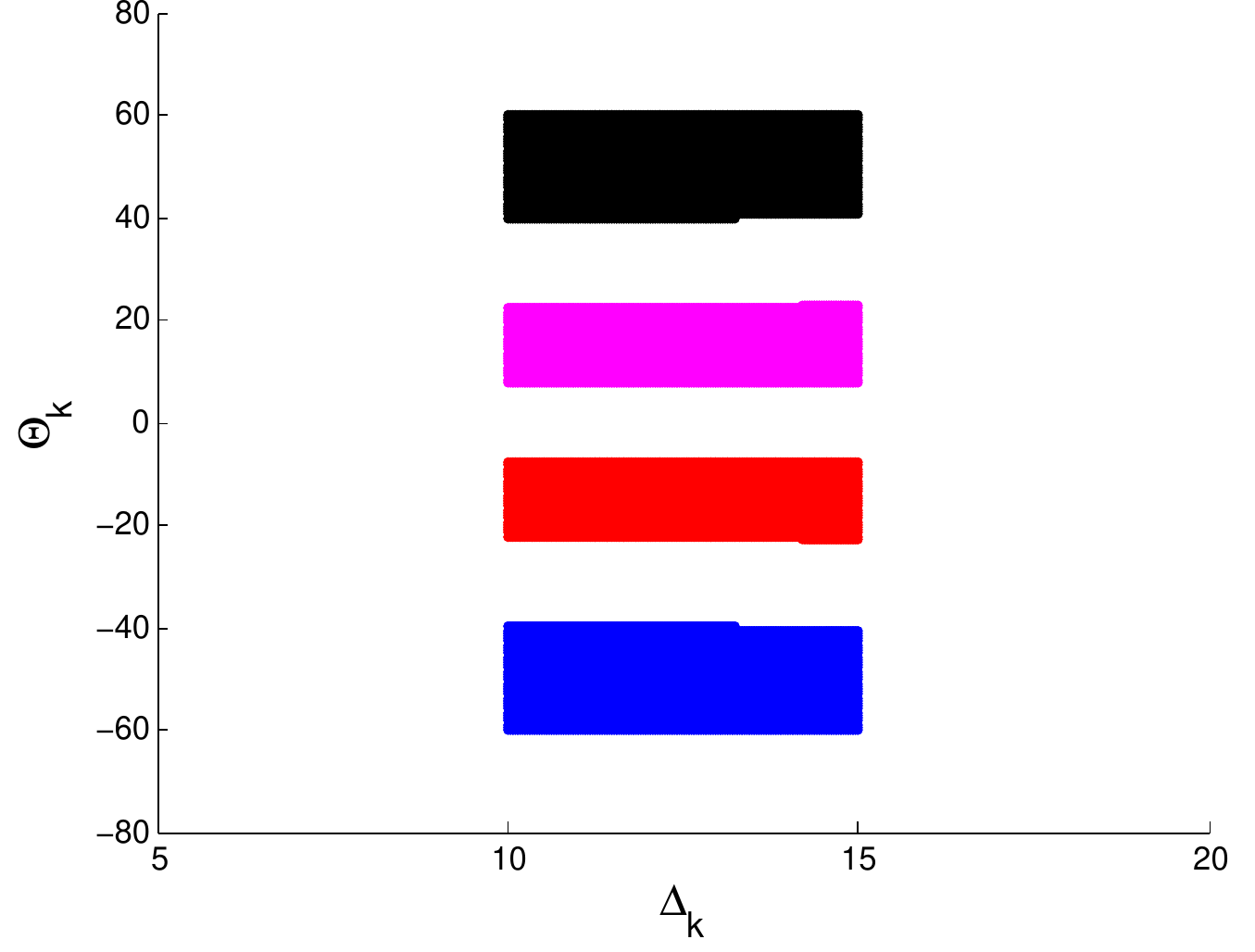}
  \label{fig:pattern-1}
  }
  \subfigure[Pattern 2]{
  \includegraphics[width=8cm]{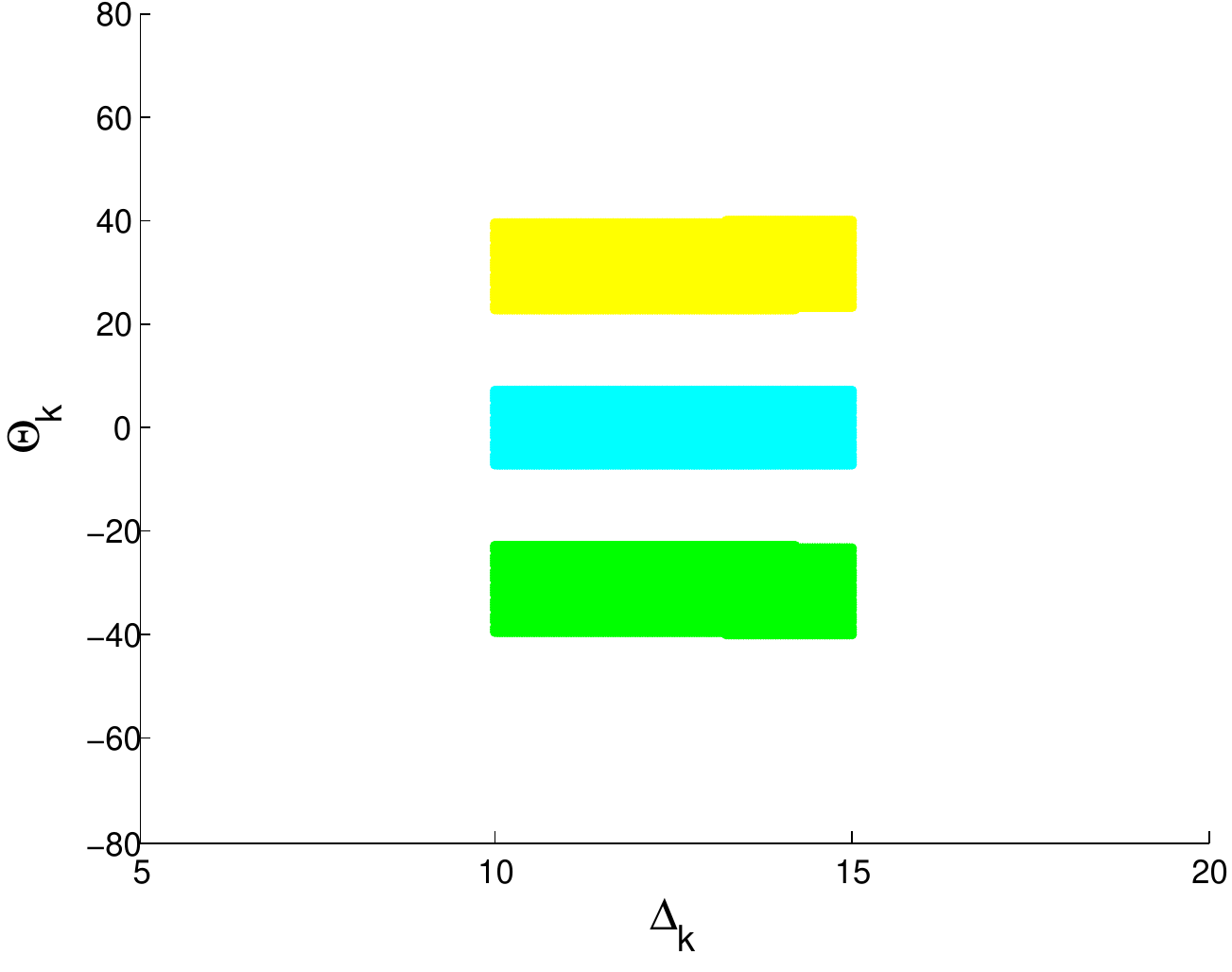}
  \label{fig:pattern-2}
  }
  \caption{Partition of the $\theta-\Delta$ plane into different patterns.
  Within each pattern, there are different groups.}
  \label{fig:patterns}
\end{figure}

We present some numerical results demonstrating the performance of the simplified user grouping algorithm based on quantization of the AoA-AS plane
in conjunction with the proposed probabilistic user selection, for user fractions obtained by greedy optimization as seen before, for different network utility functions $\Gc(\cdot)$. Specifically, we focus on two cases: 1) Proportional fairness scheduling (PFS), corresponding to the choice $\Gc(\bar{R}_1,\ldots,\bar{R}_K) = \sum_k \log \bar{R}_k$; 2)
Sum rate maximization, corresponding to the choice $\Gc(\bar{R}_1,\ldots,\bar{R}_K) = \sum_k \bar{R}_k$.
We assume a uniform distribution for the users' angle of arrival $\theta_k \in (-60^o,60^o)$, and angular spread $\Delta_k \in (5^o,15^o)$, set the number of groups equal to 8 and divide these user groups into two overlapping patterns containing $G = 4$ groups each. Pattern 1 contains the groups 1,3,5 and 7, and pattern 2 contains the groups 2,4,6 and 8. Similar to Example \ref{example:user-grouping}, for pattern 1, we have $A_g = \frac{g-\frac{1}{2}}{G} - \frac{1}{2}$ where $g \in \{1,2,3,4\}$. For pattern 2, we have $A_g = \frac{g}{G} - \frac{1}{2}$ and $g \in \{1,2,3,4\}$.
We partition the user population based on their angles of arrival and angular spreads using the simplified user grouping algorithm described
in Section \ref{sec:user-grouping-simplified}. Figures \ref{fig:pattern-1} and \ref{fig:pattern-2} show the quantization regions in the AoA-AS plane.
After solving (off-line) for the optimal user fractions, we apply the probabilistic user selection scheme of
Section \ref{sec:prob-user-selection} in order to schedule the users within each pattern.
The two patterns are served in orthogonal time-frequency slots, with equal sharing of the transmission resource.

\begin{figure}
\centering \subfigure[Objective vs. $S$]{
  \includegraphics[width=8cm]{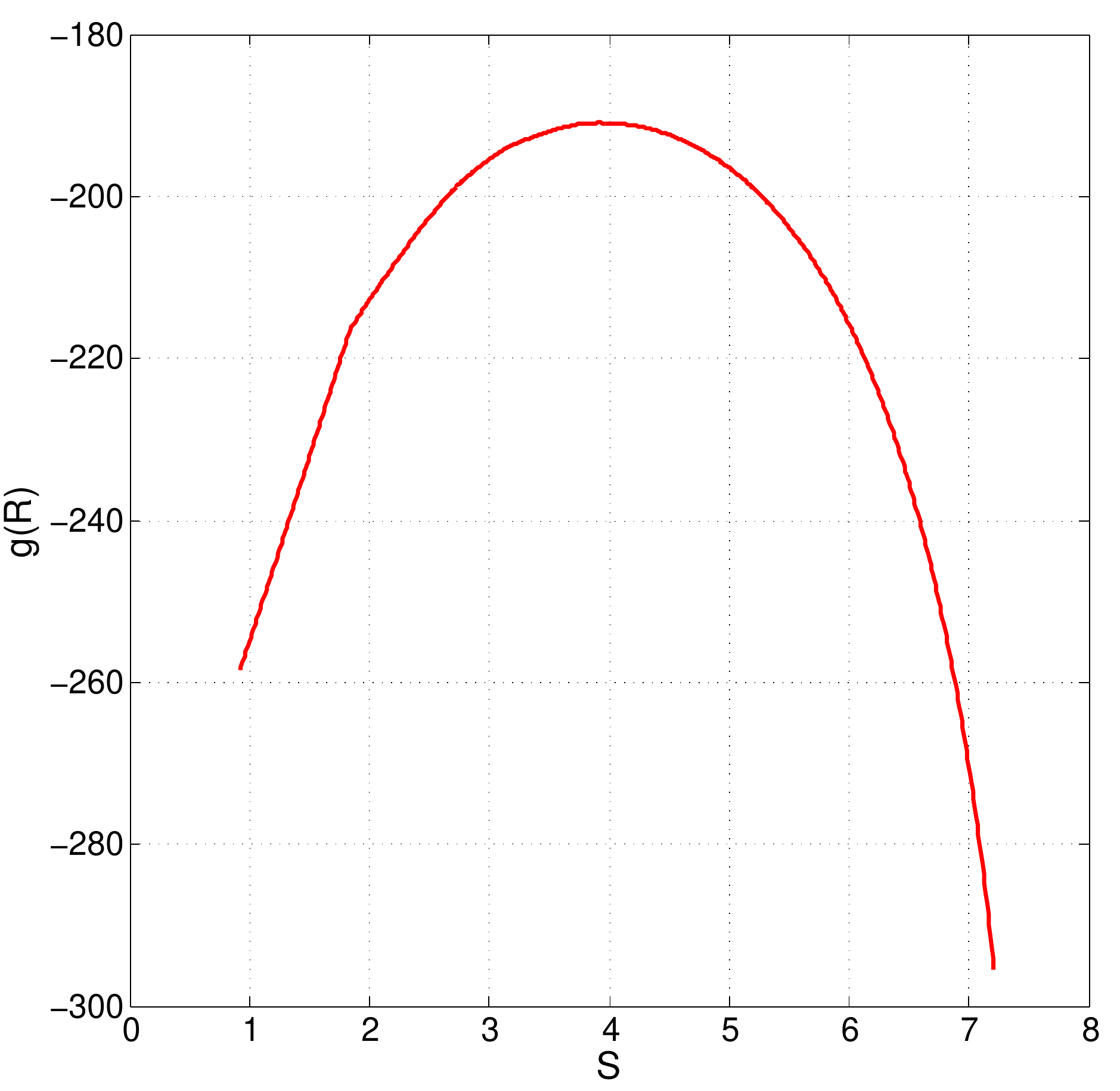}
  \label{fig:gamma-vs-obj-pfs}
  }
  \subfigure[Rate distribution]{
  \includegraphics[width=8cm]{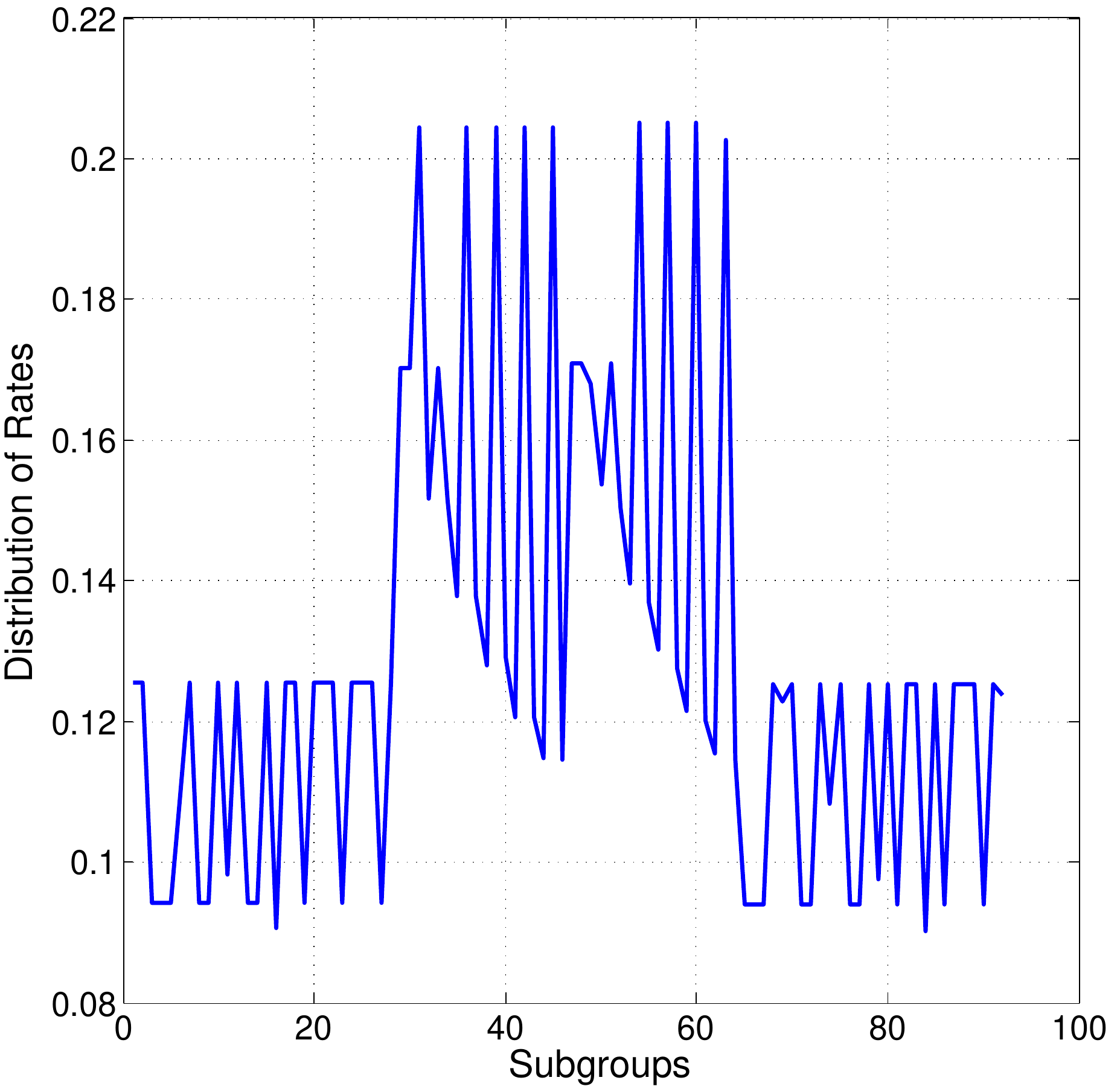}
  \label{fig:rate-dist-pfs}
  }
  \caption{Optimization of user subgroups fractions for proportional fairness scheduling
  in the large system limit, for Pattern 1. $G = 4, b_1 = b_2 = b_3 = b_4 = 2$, $\delta \gamma = 0.01$ and $P = 10$ dB.}
  \label{fig:pfs-example}
\end{figure}

\begin{figure}
\centering \subfigure[Objective vs. $S$]{
  \includegraphics[width=8cm]{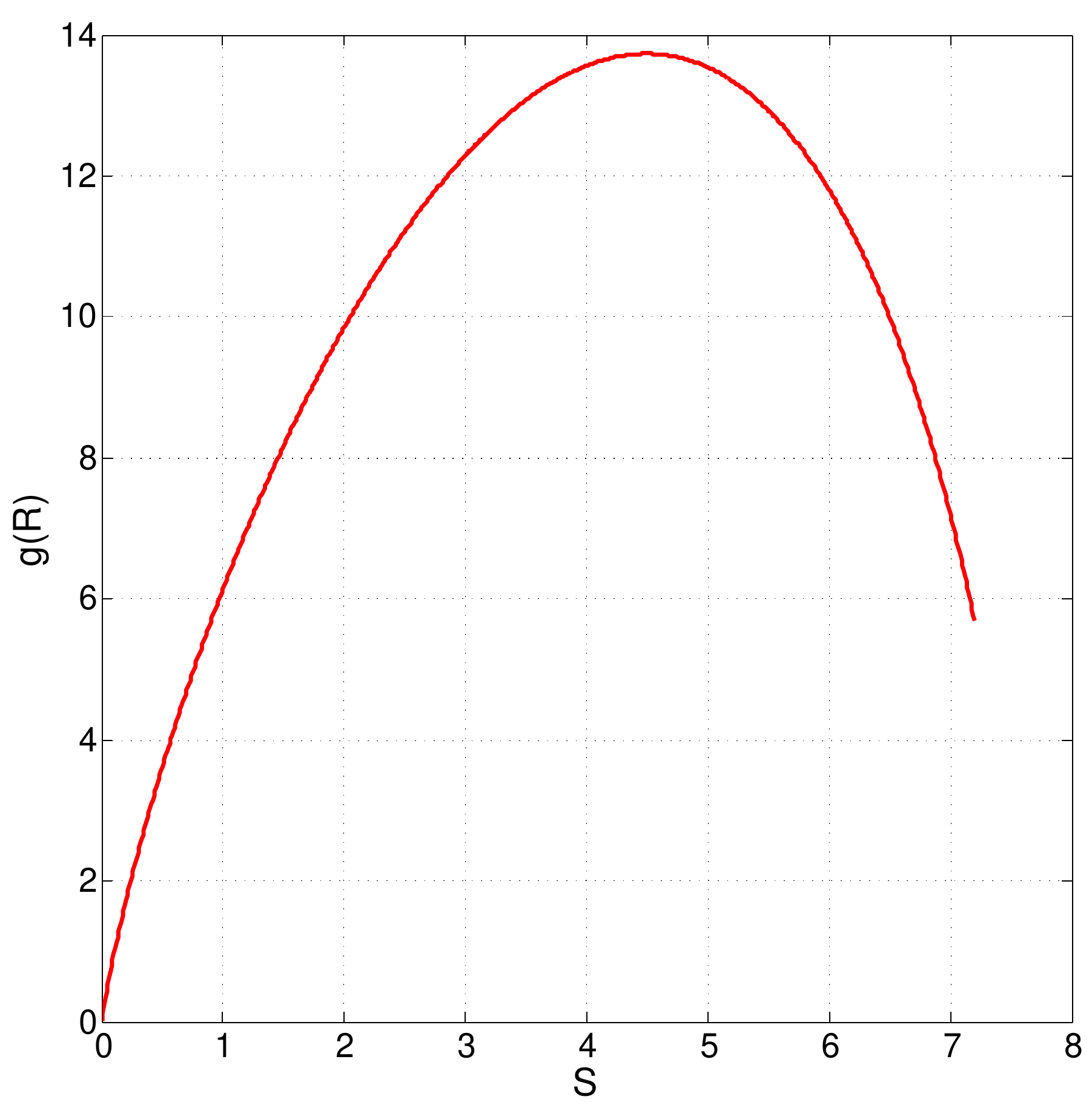}
  \label{fig:gamma-vs-obj-sr}
  }
  \subfigure[Rate distribution]{
  \includegraphics[width=8cm]{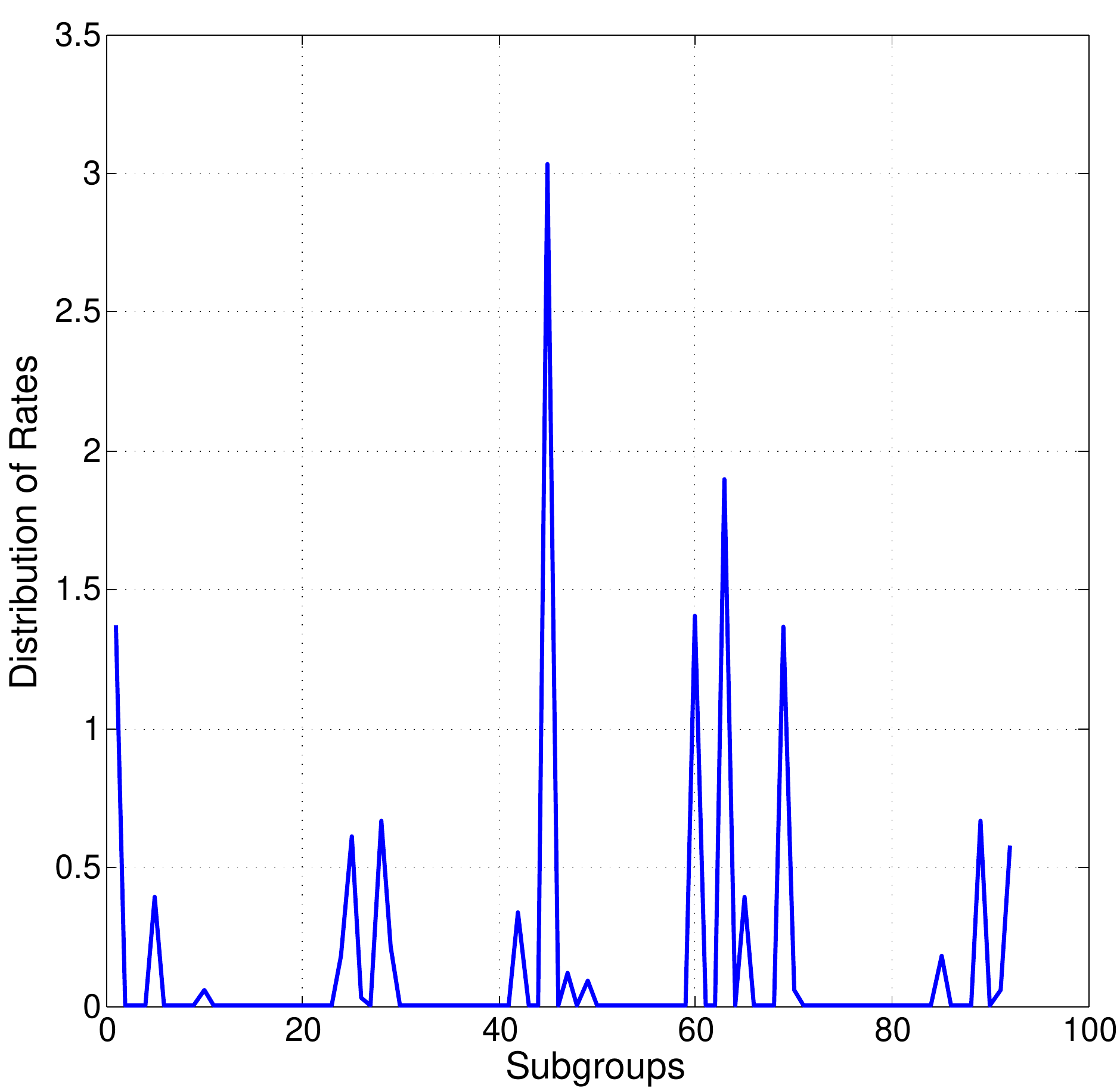}
  \label{fig:rate-dist-sr}
  }
  \caption{Optimization of user subgroups fractions for sum rate maximization in the large system limit, for Pattern 1.
  $G = 4,b_1 = b_2 = b_3 = b_4 = 2$, $\delta \gamma = 0.01$ and $P = 10$ dB.}
  \label{fig:sr-example}
\end{figure}

Figures \ref{fig:gamma-vs-obj-pfs} and \ref{fig:gamma-vs-obj-sr} shows the network utility objective function versus
$S = \sum_{g = 1}^G \sum_{k=1}^{K_g} \gamma_{g_k}$ for proportional fairness  and sum rate maximization for Pattern 1, respectively.
In this example we have $M = 8, G = 4,b_1 = b_2 = b_3 = b_4 = 2$, $\delta \gamma = 0.01$ and $P = 10$ dB.
The optimization is performed by applying the greedy heuristic algorithm while omitting Step 4, in order to find the value of the objective function for increasing $S$ even beyond its maximum, for the sake of illustration. In this case, we terminate the algorithm when no pair $(g,k)$ can be found such that $\gamma_{g_k} \leq 1$
and $\sum_{k  = 1}^{K_g} \gamma_{g_k} \leq b_g\ \ \forall g$.
Figures \ref{fig:rate-dist-pfs} and \ref{fig:rate-dist-sr} show the distribution of the rates in different subgroups under the two considered
network utility functions. In these figures, we plot the normalized rates corresponding to a subgroup versus the subgroup index for pattern 1. We notice that the user rate distribution is fair under PFS whereas for sum rate maximization
only a few subgroups have positive rates, leaving many other users completely starving.

In a practical finite-dimensional system, for given user fractions $\{\gamma_{g_k}\}$,  the users to be scheduled are selected randomly
in the following manner: the BS can transmit a maximum of $b_g N$ independent data streams in each group $g$.
At each slot, within each group $g$, the BS generates $b_g N$ i.i.d. random variables $X_1,\ldots,X_{b_g N}$ taking on values from the set of
integers $\{0,1,\ldots,b_g\}$ such that $\PP (X_{m} = k) = \frac{\gamma_{g_k}}{b_g} \ \forall \ k \neq 0$ and $\PP(X_{m} = 0) = 1 - \sum_{k=1}^{K_g} \frac{\gamma_{g_k}}{b_g}$.
A user in the $k$-th  subgroup of group $g$ is then served by the $m$-th downlink  stream on the current time-frequency
slot if $X_m = k$. The next few results demonstrate the effectiveness of the simplified  user grouping, greedy heuristic for optimization of user fractions and corresponding probabilistic selection.


\begin{figure}
\centering \subfigure[PFS]{
  \includegraphics[width=8cm]{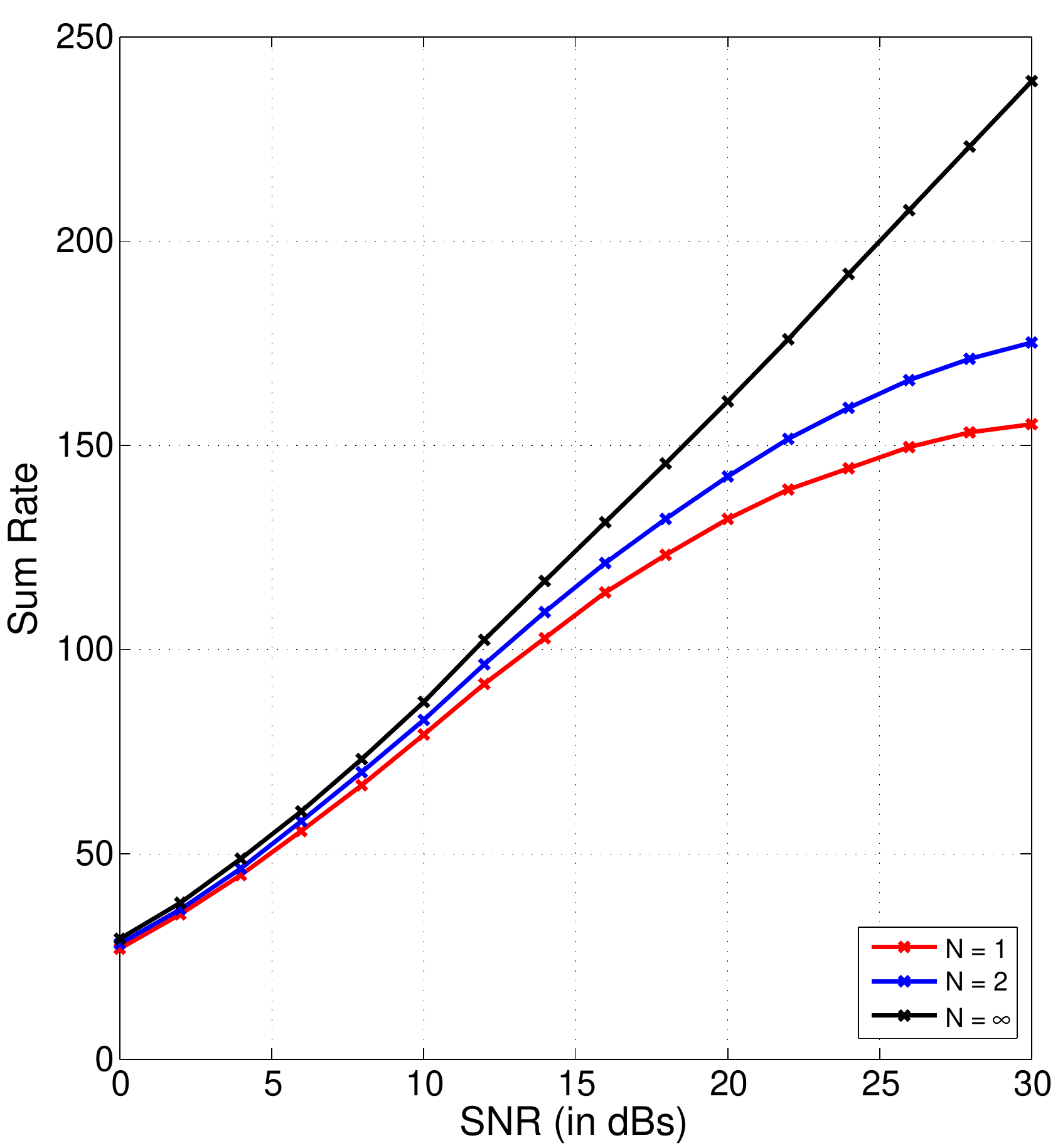}
  \label{fig:M-64-pfs}
  }
  \subfigure[Sum Rate Maximization]{
  \includegraphics[width=8cm]{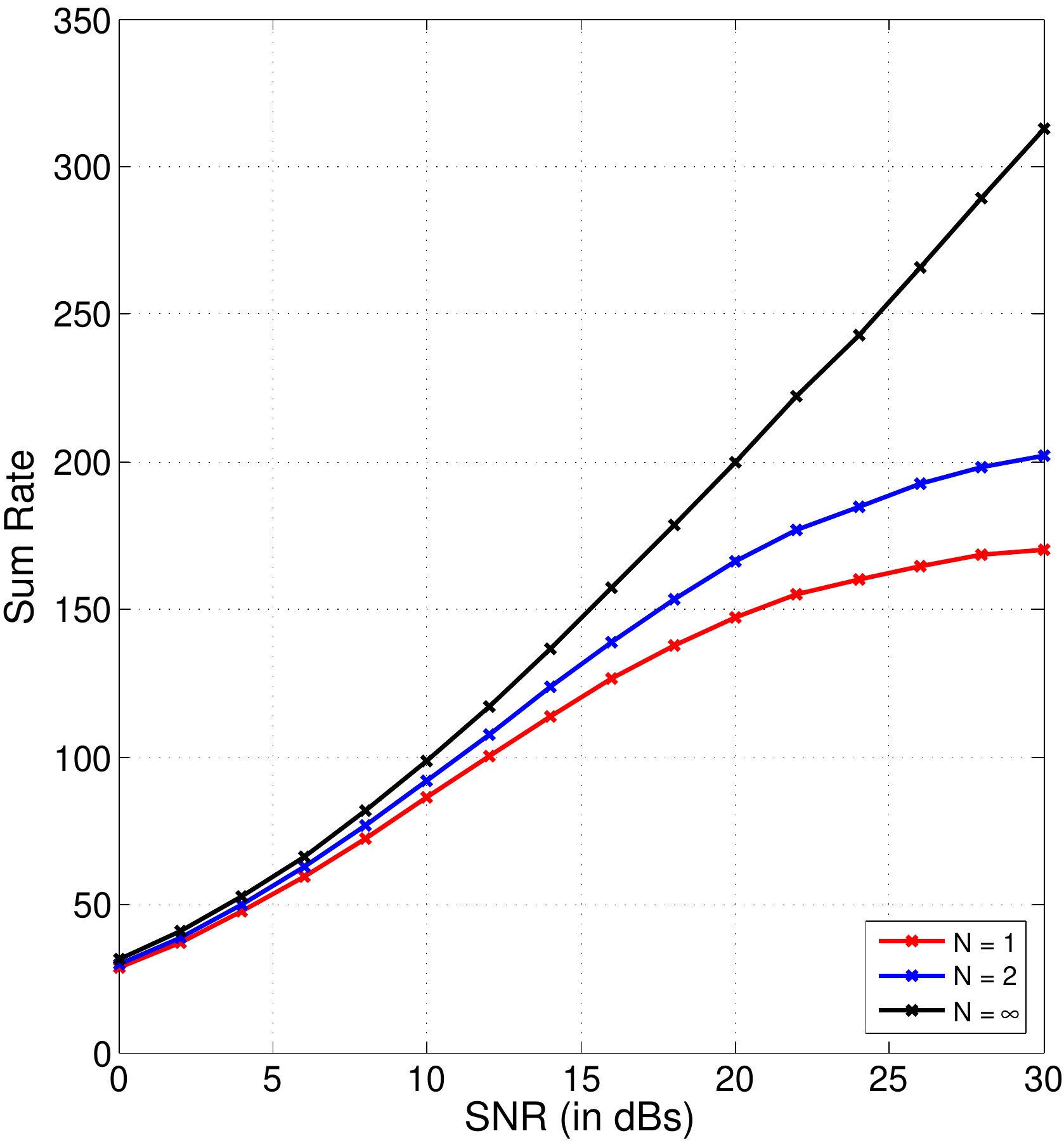}
  \label{fig:M-64-sum}
  }
  \caption{Comparison of sum spectral efficiency (bit/s/Hz) vs. SNR for JSDM with $M = 64$ and varying $N$ for simplified user grouping and probabilistic user scheduling
  with different fairness functions.}
\end{figure}

\begin{figure}
\centering \subfigure[PFS]{
  \includegraphics[width=8cm]{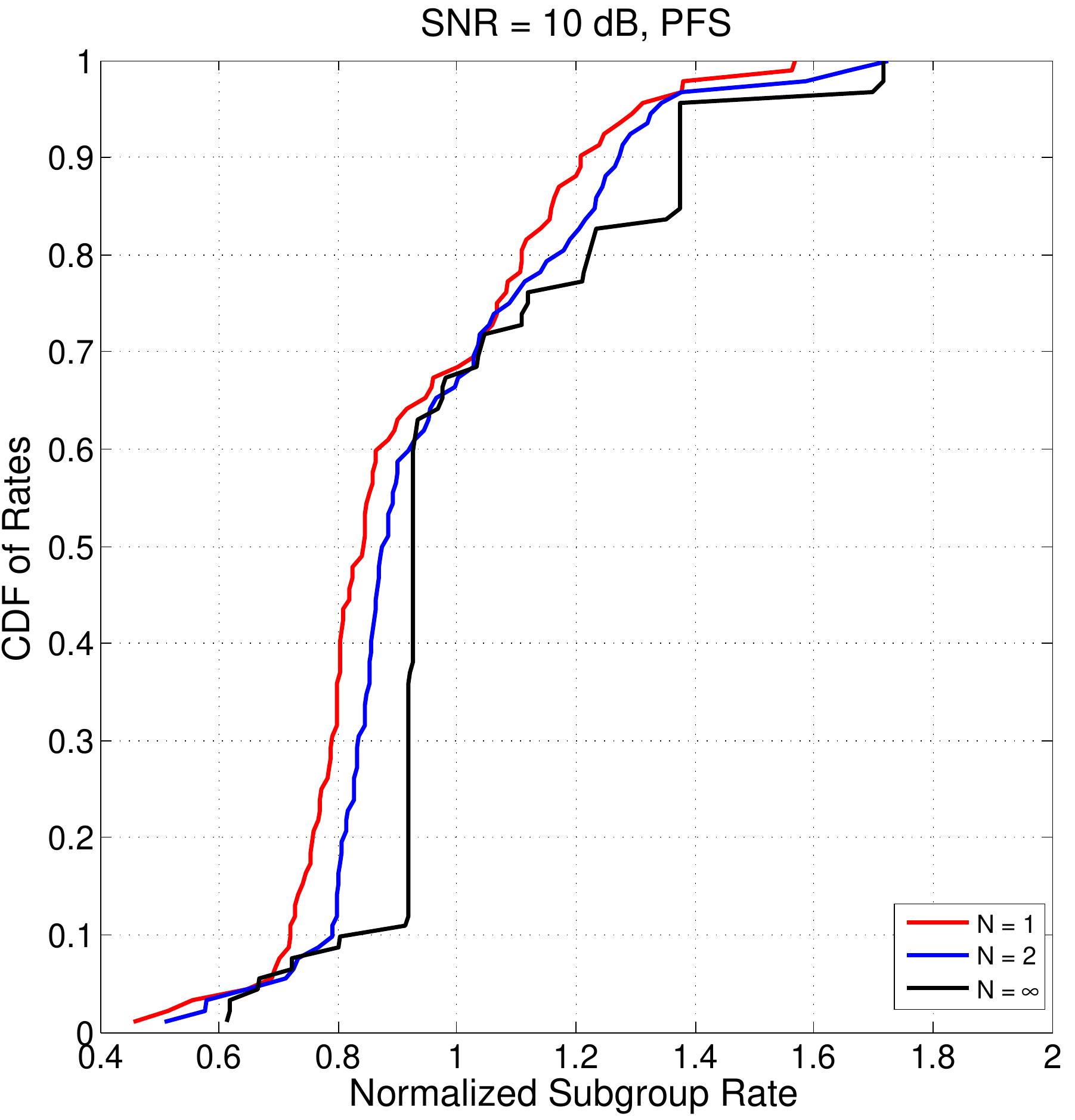}
  \label{fig:M-64-pfs-rate-dist}
  }
  \subfigure[Sum Rate Maximization]{
  \includegraphics[width=8cm]{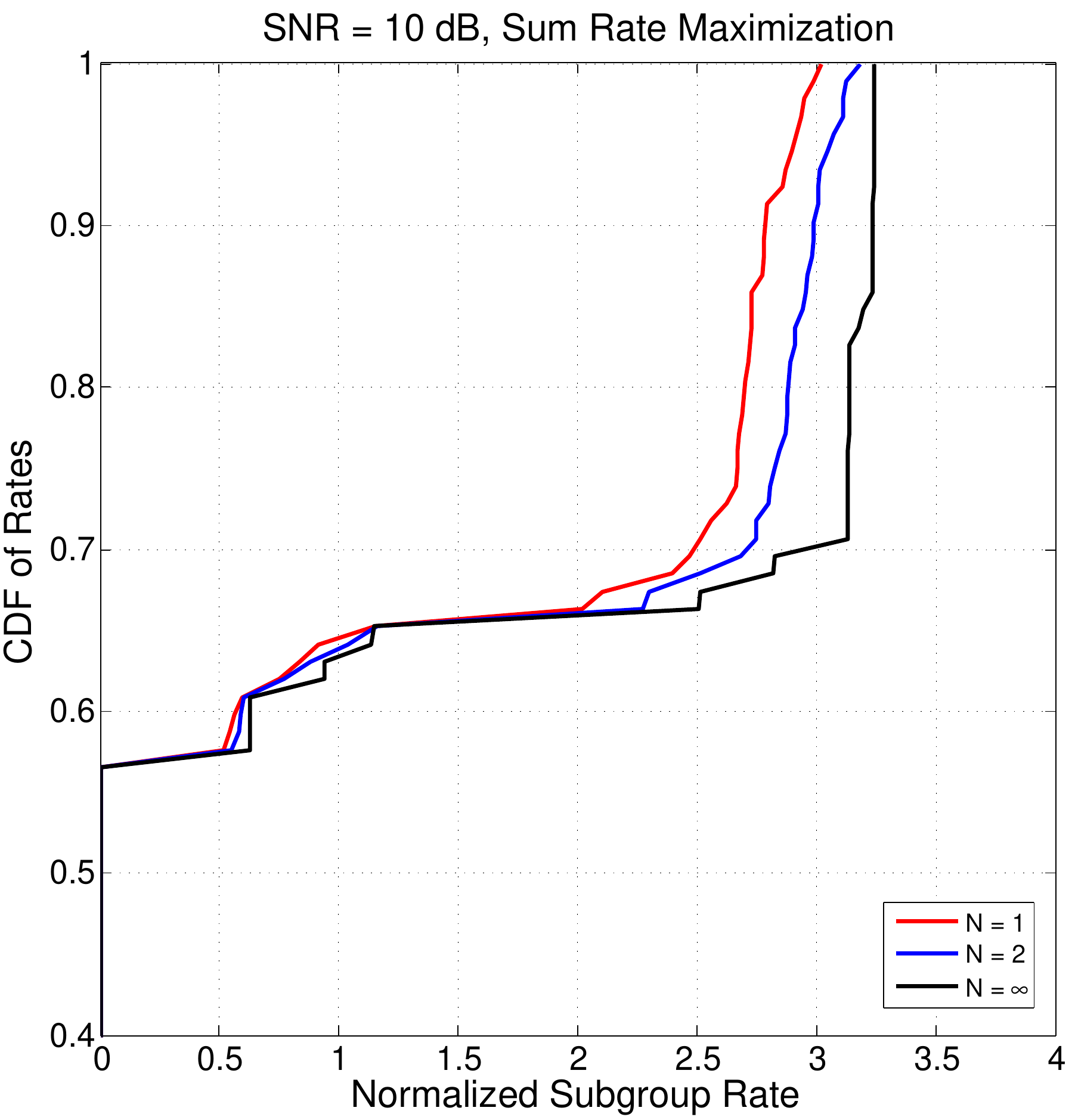}
  \label{fig:M-64-sum-rate-dist}
  }
  \caption{CDFs of the normalized subgroup rates for JSDM with $M = 64$, ${\rm SNR} = 10\ dB$ and varying $N$ for simplified user grouping and probabilistic user scheduling
  with different fairness functions.}
\end{figure}

Figures \ref{fig:M-64-pfs} and \ref{fig:M-64-sum} show the sum rate obtained for PFS and Sum rate Maximization,
when simplified user grouping algorithm is applied and the optimal user fractions are obtained using the greedy heuristic based algorithm of Section \ref{sec:prob-user-selection}. The ``sum rate'' refers to the normalized sum rate averaged over the patterns. We fix $M = 64$ and compare the finite dimensional simulations (obtained for $N = 1$ and $N = 2$ and denoted by the ``red'' and ``blue'' curves) with the large system approximations (shown by the ``black'' curve). The finite dimensional simulations differ from those obtained using the large system results because of the intergroup interference, which does not vanish for finite $N$. With increasing $N$, the finite dimensional results will ultimately coincide with the large system limit.
Figures \ref{fig:M-64-pfs-rate-dist} and \ref{fig:M-64-sum-rate-dist} show the cumulative distribution of the normalized subgroup rates for $M = 64$ for a fixed SNR = 10 dB, with varying $N$. It is apparent that as $N$ increases, the distribution of the normalized rates for the subgroups approaches to that obtained from the large system analysis.
We observe that, as expected, the group rates in the case of PFS are all positive, indicating that groups are served with some fairness.
Instead, the group rate CDF for the case of sum rate maximization shows a ``jump'' at zero, indicating the fraction of groups that are given zero rate.
In this case, the users in these groups are not served at all, and the system is unfair in favor of a higher total throughput.

Also, as already noticed before, we wish to stress the fact that the proposed probabilistic user selection scheme
involves a reduced channel state information feedback with respect to the standard greedy user selection that needs
all users to feed back their effective channels. For example, user selection based methods proposed in Section \ref{sec:sim-finite} require
feedback of the order of the total number of users in the system, whereas the proposed scheme requires feedback only from a subset
of users (the size of this subset is always less than the number of spatial dimensions available for multiplexing), that are pre-selected
based on the user fractions computed using approximations in the large system limit.

\section{Conclusion} \label{sec:conclusions}

JSDM is a multiuser MIMO downlink scheme that aims to serve users by clustering them into groups such that users within a group have approximately similar channel covariances, while users across groups have near orthogonal covariances.
JSDM was proposed in \cite{ciss2012} and analyzed in the large system limit in \cite{adhikary2012joint} under the assumption that the user channel covariance matrices
are grouped into sets with exact the same eigenspace. In this paper, we have significantly extended these results in two ways.
For the case of a finite number of BS antennas and large number of users, we obtained the scaling laws of the system
sum capacity and showed that the sum capacity scales as $\log \log K$, where $K$ is the number of users in the system, with a coefficient that
depends on the sum of the ranks of the user group covariance matrices. Then, we investigated the general problem of clustering the users
into groups (user grouping) when, realistically, each user has its own individual channel covariance matrix (i.e., no a priori groups with same covariance matrix are assumed).
We  proposed a simplified algorithm requiring only the knowledge of the users AoA and AS  (i.e., the angular support of the scattering from which the BS transmit power is
received at the user antenna). The proposed simplified grouping corresponds to the quantization or the AoA-AS plane and
works well when the number of BS antennas is large. Finally, we considered the performance analysis in the large system limit (both large number of users and
large number of BS antennas), obtained appealing closed-form fixed-point equations that enable to calculate the SINR for each user, and based on these expression, we
proposed a method to optimize the number of downlink streams to be served by JSDM for each (discretized) point in the AoA-AS plane.
This can be optimized depending on a desired network utility function of the user rates, which can be chosen to implement a desired notion of fairness.
Based on this optimization, we also proposed a probabilistic user selection that implicitly allocates the number of streams to
the users according to the optimal downlink stream distribution.  Finite dimensional simulations show the effectiveness of the proposed method.

\appendix \label{sec:appendix}

\subsection{Computation of $\PP \left ( {\rm SINR}_{k,m} > x \right )$}
\label{subsec:cdf}

Denoting by $u(Z)$ the unit step function of $Z$, we have
\begin{eqnarray}
\PP \left( {\rm SINR}_{k,m} > x \right) &=& \PP ( Z > 0 ) \nonumber\\
&=& \int_{-\infty}^{\infty} u(Z) \mathbf{f} (\wv_{g_k}) \dv
\wv_{g_k}
\nonumber\\
&=& \frac{1}{\pi^{r_g}} \int_{-\infty}^{\infty} u(Z)
e^{-||\wv_{g_k}||^2} \dv \wv_{g_k}
\nonumber\\
&=& \frac{1}{\pi^{r_g}} \frac{1}{2\pi j} \int_{-\infty}^{\infty}
\int_{-\infty}^{\infty} \frac{e^{(j\omega + c)Z}}{j\omega +
c} e^{-||\wv_{g_k}||^2} \dv \wv_{g_k} d \omega
\nonumber\\
&=& \frac{1}{\pi^{r_g}} \frac{1}{2\pi j} \int_{-\infty}^{\infty}
\frac{e^{-(j\omega + c)x/\rho}}{j\omega + c} d\omega
\int_{-\infty}^{\infty} e^{-\wv_{g_k}^\herm [\Id_{r_g} - (j\omega +
c)(\Am_{m_1} - x\Am_{m_2})] \wv_{g_k}} \dv \wv_{g_k}
\nonumber\\
&=& \frac{1}{2\pi j} \int_{-\infty}^{\infty} \frac{e^{-(j\omega +
c)x/\rho}}{j\omega + c} \frac{1}{\det{(\Id_{r_g} - (j\omega
+ c)(\Am_{m_1} - x\Am_{m_2}))}} d\omega
\nonumber\\
&=& \frac{1}{2\pi j} \int_{-\infty}^{\infty} \frac{e^{-(j\omega +
c)x/\rho}}{j\omega + c} \frac{1}{(\prod_{i=1}^{r_g}1 -
(j\omega + c)\mu_{i,m}(x))} d\omega
\nonumber\\
&\stackrel{(a)}{=}& \sum_{k: \mu_{k,m} > 0} \left. \frac{e^{-(j\omega +
c)x/\rho}}{\mu_{k,m}(x) (j\omega + c)} \frac{1}{(\prod_{i=1,i \neq k}^{r_g}1 -
(j\omega + c)\mu_{i,m}(x))} \right|_{j\omega + c = \frac{1}{\mu_{k,m}(x)}}
\end{eqnarray}
where $(a)$ follows from invoking Cauchy's integral theorem.

\subsection{Extreme Value Theory} \label{subsec:extreme-val-theory}

For the sake of completeness, we recall here some known results on the asymptotic behavior of the maximum of $K'$ random variables
as $K' \rightarrow \infty$ (see \cite{sharif2005capacity},\cite{david1970order} and references therein).
For an arbitrary distribution, the density of the maximum does not necessarily have a limit as $K'$ goes
to infinity. (Gnedenko, 1947) lists all possible limiting distributions for the cumulative distribution of the maximum of $K'$ i.i.d. random variables.
\begin{thm}[Gnedenko,1947]
Let $x_1,\ldots,x_{K'}$ denote a sequence of i.i.d. random variables with
$$x_{\max} = \max(x_1,\ldots,x_{K'})$$
Suppose that for some sequences, $\{a_n > 0\}$, ${b_n}$ of real constants, $a_n(x_{\max} - b_n)$ converges in distribution to a random variable with distribution function $G(x)$. Then, $G(x)$ must be one of the following three types:
\begin{enumerate}
\item $G(x) = e^{e^{-x}}$
\item $G(x) = e^{-x^{-\alpha}} u(x), \alpha > 0$
\item $G(x) = \left\{ \begin{array}{cc} e^{-(-x)^{\alpha}} & \alpha > 0, x \leq 0\\
1 & x > 0\end{array} \right.$
\end{enumerate}
where $u(x)$ denotes the unit step function.
\end{thm}

The class of distribution functions in our work are of Type 1. Denoting by $f_X(x)$ and $F_X(x)$ the probability density and probability
distribution functions of a random variable $x$ for $x > 0$, define the growth function as $g(x) = \frac{1 - F_X(x)}{f_X(x)}$.
Also, let $u_{K'}$ be the unique solution to
$$1 - F_X(u_{K'}) = \frac{1}{K'}$$. We have the following result:
\begin{thm}[Uzgoren, 1956]
Let $x_1,\ldots,x_{K'}$ denote a sequence of i.i.d. positive random variables with $f_X(x) > 0$ and growth function $g(x)$.
$$\log\{-\log\{F^{K'}(u_{K'} + u g(u_{K'}))\}\} = -u - \frac{u^2 g'(u_{K'})}{2 !} - \ldots - \frac{u^m g^{(m-1)}(u_{K'})}{(m-1) !} + O\left(\frac{e^{-u + u^2 g'(u_{K'})}}{K'}\right)$$
\end{thm}

\subsection{Proof of (\ref{eqn:limit-growth})} \label{subsec:compute-limit-growth}

We have
\begin{eqnarray}
g(x) &=& \frac{1 - F_X(x)}{f_X(x)} \nonumber\\
&=& \frac{1 - F_X(x)}{F'_X(x)} \nonumber\\
&=& \frac{1}{\frac{1}{\rho \mu_{1,m}(x)} - \frac{x \mu'_{1,m}(x)}{\rho (\mu_{1,m}(x)})^2 + \sum_{i=2}^{r_g} \frac{\mu_{i,m}(x) \mu'_{1,m}(x) - \mu_{1,m}(x) \mu'_{i,m}(x)}{(\mu_{1,m}(x) - \mu_{i,m}(x)) \mu_{1,m}(x)}}
\end{eqnarray}
As $x \rightarrow \infty$, $\mu_{1,m}(x) \rightarrow \mu_{1,m}^*$, $x \mu'_{1,m}(x) \rightarrow 0$, $\frac{\mu_{i,m}(x) \mu'_{1,m}(x)}{\mu_{1,m}(x) - \mu_{i,m}(x)} \rightarrow 0$ and $\frac{\mu'_{i,m}(x) \mu_{1,m}(x)}{\mu_{1,m}(x) - \mu_{i,m}(x)} \rightarrow 0$, giving
$$\lim_{x \rightarrow \infty} g(x) = \frac{1}{\frac{1}{\rho \mu_{1,m}^*}} = \rho \mu_{1,m}^*$$

\bibliographystyle{IEEEtran}
\bibliography{references}

\end{document}

%% file: macros.tex
\long\def\comment#1{}

\newfont{\bbb}{msbm10 scaled 700}

\newfont{\bbc}{msbm10 scaled 1100}
\newcommand{\CC}{\mbox{\bbc C}}
\newcommand{\PP}{\mbox{\bbc P}}
\newcommand{\RR}{\mbox{\bbc R}}

\newcommand{\EE}{\mbox{\bbc E}}

\newcommand{\bvx}{{\pmb b}}

\newcommand{\dv}{{\pmb d}}
\newcommand{\ev}{{\pmb e}}
\newcommand{\fv}{{\pmb f}}

\newcommand{\hv}{{\pmb h}}

\newcommand{\kv}{{\pmb k}}

\newcommand{\nv}{{\pmb n}}

\newcommand{\qv}{{\pmb q}}

\newcommand{\uv}{{\pmb u}}
\newcommand{\vv}{{\pmb v}}
\newcommand{\wv}{{\pmb w}}
\newcommand{\xv}{{\pmb x}}
\newcommand{\yv}{{\pmb y}}
\newcommand{\zv}{{\pmb z}}
\newcommand{\zerov}{{\pmb 0}}

\newcommand{\Am}{{\pmb A}}
\newcommand{\Bm}{{\pmb B}}

\newcommand{\Fm}{{\pmb F}}

\newcommand{\Hm}{{\pmb H}}
\newcommand{\Id}{{\pmb I}}
\newcommand{\Jm}{{\pmb J}}
\newcommand{\Km}{{\pmb K}}

\newcommand{\Pm}{{\pmb P}}
\newcommand{\Qm}{{\pmb Q}}
\newcommand{\Rm}{{\pmb R}}

\newcommand{\Tm}{{\pmb T}}
\newcommand{\Um}{{\pmb U}}
\newcommand{\Vm}{{\pmb V}}
\newcommand{\Wm}{{\pmb W}}
\newcommand{\Xm}{{\pmb X}}
\newcommand{\Ym}{{\pmb Y}}

\newcommand{\Cc}{{\cal C}}

\newcommand{\Fc}{{\cal F}}
\newcommand{\Gc}{{\cal G}}

\newcommand{\Nc}{{\cal N}}

\newcommand{\Rc}{{\cal R}}
\newcommand{\Sc}{{\cal S}}

\newcommand{\gammav}{\hbox{\boldmath$\gamma$}}

\newcommand{\muv}{\hbox{\boldmath$\mu$}}

\newcommand{\Lambdam}{\hbox{\boldmath$\Lambda$}}

\newcommand{\Sigmam}{\hbox{\boldmath$\Sigma$}}

\renewcommand{\det}{{\hbox{det}}}
\newcommand{\trace}{{\hbox{tr}}}

\newcommand{\eqdef}{\stackrel{\Delta}{=}}

\newcommand{\herm}{{\sf H}}
